\documentclass[a4paper]{article}

\usepackage{hyperref}

\hypersetup{
    colorlinks=true,
    linkcolor=darkblue,
    citecolor=darkblue,
    filecolor=darkblue,
    urlcolor=darkblue,
    bookmarks=true,
    pdfpagemode=FullScreen
}

\usepackage{amsmath,amssymb,amsthm,mathabx,todonotes,microtype,thmtools,mathtools,xspace,color,cleveref}
\usepackage[shortcuts]{extdash}
\usepackage[utf8]{inputenc}
\usepackage[many]{tcolorbox}
\usepackage{enumitem}

\title{\texorpdfstring{Deterministic and game separability\\ for regular languages of infinite trees}{Deterministic and game separability for regular languages of infinite trees}}

\author{Lorenzo Clemente \and Michał Skrzypczak}

\newtheorem{thm}{Theorem}[section]

\newcommand{\trees}{\mathrm{Tr}}

\newcommand{\ign}[1]{}

\newcommand{\N}{\mathbb N}
\newcommand{\A}{\mathcal A}
\newcommand{\B}{\mathcal B}
\newcommand{\C}{\mathcal C}
\newcommand{\D}{\mathcal D}

\newcommand{\M}{\mathcal M}
\newcommand{\PP}{\mathcal P}
\renewcommand{\S}{\mathcal S}
\newcommand{\T}{\mathcal R}
\newcommand{\WW}{\mathcal W}

\newcommand{\ignore}[1]{}

\newcommand{\Win}[1]{\mathbf{W}_{#1}}

\newcommand{\Choice}[2]{[\ensuremath{\textsf{#1}\colon{#2}}]}

\newcommand{\GChoice}[3]{\ensuremath{{#1}.[\textsf{#2}\colon{#3}]}}

\newcommand{\paths}[1]{\mathsf{Path}(#1)}
\newcommand{\pathclosure}[1]{\forall\mathsf{Path}(#1)}
\newcommand{\pathautomaton}[1]{{#1}^{\mathrm{path}}}
\newcommand{\tuple}[1]{(#1)}
\newcommand{\sep}{\;|\;}
\renewcommand{\L}{\mathsf L}
\newcommand{\R}{\mathsf R}
\newcommand{\set}[1]{\{#1\}}
\newcommand{\Set}[1]{\left\{#1\right\}}
\newcommand{\setof}[2]{\set{#1 \sep #2}}
\newcommand{\Setof}[2]{\Set{#1 \mid #2}}
\renewcommand{\Game}[1]{G_{#1}}
\newcommand{\AcceptanceGame}[2]{G^{\mathrm{acc}}(#1, #2)}

\newcommand{\DisjointnessGame}[2]{G^{\mathrm{dis}}(#1, #2)}

\newcommand{\OmegaSeparabilityGame}[3]{G^{\mathrm{sep}}_\omega(#1, #2, #3)}
\newcommand{\DeterministicSeparabilityGame}[2]{G^{\mathrm{sep}}_{\mathrm{det}}(#1, #2)}
\newcommand{\CDeterministicSeparabilityGame}[3]{G^{\mathrm{sep}}_{\mathrm{det}}(#1, #2, #3)}
\newcommand{\GameSeparabilityGame}[2]{G^{\mathrm{sep}}_{\mathrm{game}}(#1, #2)}
\newcommand{\CGameSeparabilityGame}[3]{G^{\mathrm{sep}}_{\mathrm{game}}(#1, #2, #3)}
\newcommand{\SepGame}{{\color{darkred}G_0}}

\newcommand{\AccGameA}{{\color{darkgreen}G_1}}
\newcommand{\AccGameS}{{\color{darkblue}G_2}}

\newcommand{\lang}[1]{\mathrm{L}(#1)}
\newcommand{\blang}[1]{\mathrm{L}^{\textsf{path}}(#1)}
\DeclareMathOperator{\disjoint}{\bot}

\newcommand{\card}[1]{\left|#1\right|}

\newcommand{\Automaton}{\textsf{Automaton}\xspace}
\newcommand{\Pathfinder}{\textsf{Pathfinder}\xspace}
\newcommand{\PI}{\textsf{PI}\xspace}
\newcommand{\PII}{\textsf{PII}\xspace}
\newcommand{\Input}{\textsf{Input}\xspace}
\newcommand{\Separator}{\textsf{Separator}\xspace}

\newcommand{\EXPTIME}{\textsf{EXPTIME}}

\newtheorem{theorem}[thm]{Theorem}

\newtheorem{lemma}[thm]{Lemma}

\newtheorem{claim}[thm]{Claim}

\newtheorem{remark}[thm]{Remark}
\newtheorem{proposition}[thm]{Proposition}

\crefname{fact}{Fact}{Facts}
\Crefname{fact}{Fact}{Facts}
\crefname{lemma}{Lemma}{Lemmas}
\Crefname{Lemma}{Lemma}{Lemmas}
\crefname{theorem}{Theorem}{Theorems}
\Crefname{Theorem}{Theorem}{Theorems}
\crefname{corollary}{Corollary}{Corollary}
\Crefname{corollary}{Corollary}{Corollary}
\crefname{claim}{Claim}{Claim}
\Crefname{claim}{Claim}{Claim}

\crefname{remark}{Remark}{Remarks}
\Crefname{remark}{Remark}{Remarks}

\crefname{section}{Section}{Sections}
\Crefname{section}{Section}{Sections}

\crefname{enumi}{}{}

\newif\ifstartedinmathmode
\newcommand*{\st}{
  \relax\ifmmode\startedinmathmodetrue\else\startedinmathmodefalse\fi
  \ifstartedinmathmode{\;\cdot\;}\else{s.t.~}\fi%
}

\makeatletter
\def\itemizename{itemize}
\def\enumeratename{enumerate}
\def\ifinlist{%
  \ifx\@currenvir\itemizename
    \expandafter\@firstoftwo
  \else
    \ifx\@currenvir\enumeratename
      \expandafter\@firstoftwo
    \else
      \expandafter\@secondoftwo
    \fi
  \fi}
\makeatother

\let\oldqedhere\qedhere

\renewcommand*{\qedhere}{%
  \relax\ifmmode\startedinmathmodetrue\else\startedinmathmodefalse\fi%
  \ifstartedinmathmode\tag*{\protect\oldqedhere}\else\ifinlist{\hfill\oldqedhere}{\oldqedhere}\fi%
}

\newcommand{\proje}{\mathrm{Proj}}

\definecolor{darkblue}{rgb}{0.0, 0.0, 0.75}
\definecolor{darkgreen}{rgb}{0.0, 0.55, 0.0}
\definecolor{darkred}{rgb}{0.55, 0.0, 0.0}

\newenvironment{gameize}{\begin{itemize}}{\end{itemize}}

\newenvironment{winnize}{\begin{itemize}}{\end{itemize}}

\newenvironment{simuize}{\begin{itemize}}{\end{itemize}}

\newtcolorbox{gamebox}[1]{%
    tikznode boxed title,
    enhanced,
    arc=0mm,
    interior style={white},
    attach boxed title to top center= {yshift=-\tcboxedtitleheight/2},
    fonttitle=\bfseries,
    colbacktitle=white,coltitle=black,
    boxed title style={size=normal,colframe=white,boxrule=0pt},
    title={#1}}

\begin{document}

\maketitle

\begin{abstract}
    We show that it is decidable whether two regular languages of infinite trees
    are separable by a~deterministic language, resp., a~game language.
    We consider two variants of separability,
    depending on whether the set of priorities of the separator is fixed, or not.
    In each case, we show that separability can be decided in \EXPTIME,
    and that separating automata of exponential size suffice.
    We obtain our results by reducing to infinite duration games with $\omega$-regular winning conditions
    and applying the finite\=/memory determinacy theorem of Büchi and Landweber.
\end{abstract}

\section{Introduction}
\label{sec:intro}

One of the most intriguing and motivating problems in the field of automata theory is the \emph{membership problem}.
For two fixed classes of languages $\C$ (\emph{input class}) and $\D$ (\emph{output class}),
the $\tuple{\C, \D}$-membership problem asks,
given a~representation of a~language in $\C$,
whether this language belongs to $\D$.
Among the first results of this type is the famous theorem by Schutzenberger~\cite{schutzenberger_fo} and McNaughton-Papert~\cite{mcnaughton_counter_free},
characterising, among all regular languages of finite words,
the subclass of languages that can be defined in first-order logic.


In this paper we consider the class $\C$ of regular languages of infinite trees.
While there are many equivalent automata models for this class
\===e.g.,~Muller, Rabin, and Street automata~\cite{AutomataLogicsInfiniteGames:2002}\===%
\emph{parity automata} are without doubt the most established such model~\cite{EmersonJutla:SFCS:1991}.
The most important descriptional complexity measure of a parity automaton
is the set of priorities $C \subseteq \N$ it is allowed to use,
which is called its \emph{index}.
Not only a larger index allows the automaton to recognise more languages~\cite{Niwinski:Clones:1986},
but the computational complexity of known procedures for the emptiness problem
crucially depends on the index (the current best bound is quasi-polynomial~\cite{CaludeJainKhoussainovLiStephan:STOC:2017}).
The most famous open problem in the area of regular languages of infinite trees
is the \emph{nondeterministic index membership problem},
which is the $\tuple{\C, \D}$-membership problem
for $\D$ the class of languages recognised by some nondeterministic parity automaton
of a fixed index~$C$~(c.f.~\cite{ColcombetLoding:ICALP:2008}).
In many cases, the solution of the membership problem relies either on algebraic representations
or determinisation,
however algebraic structures for regular languages of infinite trees are of limited availability~(c.f.~\cite{blumensath_recognisability})
and deterministic automata do not capture all regular languages.
While on infinite words this problem was essentially solved by Wagner already at the end of the '70s~\cite{Wagner:IC:1979},
its solution for infinite trees seems still far away.

Known decidability results abound if we restrict either the input class $\C$ or the output class $\D$.
Results of the first kind are known for $\C$ being the class of deterministic~\cite{NiwinskiWalukiewicz:ENTCS:2005}
and, more generally, game automata~\cite[Theorem~1.2]{FacchiniMurlakSkrzypczak:TOCL:2016}.
Results of the second kind (i.e., when the input class $\C$ is the full class of regular languages)
exist for the output class $\D$ being
the lower levels of the index hierarchy~\cite{KustersWilke:FSTTCS:2002,Walukiewicz:ENTCS:2002}
and of the Borel hierarchy~\cite{BojaczykCavallariPlaceSkrzypczak:LMCS:2019},
the class of deterministic languages~\cite{NiwinskiWalukiewicz:STACS:1998},
and Boolean combinations of open sets~\cite{BojanczykPlace:ICALP:2012}.
Other variants of the index membership problem are known to be decidable,
including the early result of Urbański showing that it is decidable whether a~given deterministic parity tree automaton
is equivalent to some nondeterministic Büchi one~\cite{Urbanski:ICALP:2000},
the weak alternating index problems for the class of deterministic automata~\cite{Murlak:STACS:2008}
and Büchi automata~\cite{ColcombetKuperbergLodingVandenBoom:CSL:2013,SkrzypczakWalukiewicz:ICALP:2016},
and deciding whether a given parity automaton is equivalent to some nondeterministic co-Büchi automaton~\cite{ColcombetKuperbergLodingVandenBoom:CSL:2013}.

Another problem closely related to membership is separability.
%
The $\tuple{\C, \D}$\=/\emph{separability} problem asks,
given a~pair of languages $L$, $M$ in $\C$,
whether there exists a~language $S$ in $\D$ (called a~\emph{separator})
s.t.~$L\subseteq S$ and\footnote{We write $S \disjoint M$ for $S \cap M = \emptyset$.} $S\disjoint M$.
Intuitively, a~separator $S$ provides a certificate of disjointness,
yielding information on the structure of $L$, $M$ up to some chosen granularity.
The separability problem is a~generalisation of the membership problem
if the class $\C$ is closed under complement,
since we can always take $M$ to be the complement of $L$,
in which case the only candidate for the separator is $L$ itself.
There are many elegant results in computer science, formal logic, and mathematics
showing that separators always exist.
Instances include Lusin's separation theorem in topology
(two disjoint analytic sets are always separable by a Borel set; c.f.~\cite[Theorem~14.7]{Kechris}),
a folklore result in computability theory (two disjoint co-recursively enumerable sets are separable by a recursive set),
Craig's theorems in logic (jointly contradictory first-order formulas can be separated by a formula containing only symbols in the shared vocabulary~\cite{Craig:JSL:1957})
and model theory (two disjoint projective classes are separable by an elementary class~\cite{Craig:JSL:1957});
in formal language theory, a generalisation of a theorem suggested by Tarski and proved by Rabin~\cite[Theorem 29]{Rabin:Weak:1970}
states that two disjoint Büchi languages of infinite trees are separable by a weak language
(c.f.~\cite{ArnoldSantocanale:TCS:2005}).
%

In this work we study the $\tuple{\C, \D}$-separability problems
where $\C$ is the full class of regular languages of infinite trees,
and $\D$ is one of four kinds of sub-classes thereof,
depending on whether the automaton is deterministic or game,
and depending on whether we fix a~finite index $C \subseteq \mathbb N$
or we leave it unrestricted $C = \mathbb N$.
Our main result is that all four kinds of the separability problems above are decidable and in \EXPTIME.
Moreover, we show that if a separator exists, then there is one of exponential size.

\begin{restatable}{theorem}{thmComplexity}
    \label{thm:complexity}
    The deterministic and game separability problems can be solved in \EXPTIME,
    both for a fixed finite index $C \subseteq \mathbb N$,
    and an unrestricted one $C = \mathbb N$.
    Moreover, separators with exponentially many states and polynomially many priorities suffice.
\end{restatable}



Our work is permeated by the observation that the separability problem for two languages $L$, $M$
can be phrased in terms of a game of infinite duration with an~$\omega$-regular winning condition.
In such a \emph{separability game} there are two players, \Separator trying to prove that $L, M$ are separable,
and \Input with the opposite objective.
In the simple case of $\tuple{\C, \D}$-separability
where $\C$ is the class of regular languages of $\omega$-words
and $\D$ the subclass induced by deterministic parity automata of finite index $C$,
the $i$-th round of the game is as follows:
\begin{itemize}
    \item \Separator plays a priority $c_i \in C$.
    \item \Input plays a letter $a_i$ from the finite alphabet $\Sigma$.
\end{itemize}
The resulting infinite play $\tuple{c_0, a_0} \tuple{c_1, a_1}\cdots$ is won by \Separator
if 1) $a_0a_1 \cdots \in L$ implies $c_0c_1\cdots$ is accepting and
2) $a_0a_1 \cdots \notin L$ implies $c_0c_1\cdots$ is rejecting.
Since the winning condition is $\omega$-regular,
by the result of Büchi and Landweber~\cite{BuchiLandweber:AMS:1969}
we can decide who wins the game
and moreover finite-memory strategies for \Separator suffice.
Thanks to a correspondence between such strategies and deterministic separators,
\Separator wins such a game
iff there exists a deterministic automaton with priorities in $C$ separating $L$, $M$.
This provides both decidability of the separability problem and an upper-bound on the size of separators.
We design analogous games with $\omega$-regular winning conditions for the more involved case of infinite trees
for the separability problems mentioned above and apply \cite{BuchiLandweber:AMS:1969}.

The separability problems we consider have been open so far
and generalise the corresponding membership problems.
A solution for deterministic separability can easily be derived from~\cite{NiwinskiWalukiewicz:TCS:2003},
however our techniques based on games are novel and provide a unified view on all problems.
When instantiated to the specific case of membership,
our decidability results generalise the deterministic case
(for both fixed and unconstrained index)
\cite{NiwinskiWalukiewicz:TCS:2003,NiwinskiWalukiewicz:STACS:1998}
and the game membership case for unconstrained index~\cite[Theorem~7.12]{FacchiniMurlakSkrzypczak:TOCL:2016}.
We believe the game approach is much more direct than the combinatorial and pattern-based techniques used in the previous solutions, c.f.~\cite[Section~7, pp.~29--37]{FacchiniMurlakSkrzypczak:TOCL:2016}.
The game membership problem for a~fixed index $C$ has been open so far.
%

We are not aware of computation complexity results for separability problems over regular languages of infinite trees,
neither of an analysis of the size of separators.
Regarding deterministic membership,
\EXPTIME-completeness is known~\cite[Corollary~11]{NiwinskiWalukiewicz:TCS:2003},
as well as \EXPTIME~upper~\cite[end of page 12]{NiwinskiWalukiewicz:STACS:1998}
and lower bounds~\cite[Theorem~4.1]{Walukiewicz:ENTCS:2002} (c.f.,~also~\cite{KustersWilke:FSTTCS:2002})
for computing the optimal deterministic index.
Devising non-trivial complexity lower bounds for the separability problem is left for future work,
as well as extending our approach to other classes of separators.


\subparagraph{Related works.}

Over finite words,
variants of the $\tuple{\C, \D}$-separability problem have been studied for classes $\C$ both more general than the regular languages,
such as the context free languages~\cite{CzerwinskiMartensVanRooijenZeitoun:FCT:2015,Zetzsche:ICALP:2015}
and higher-order languages~\cite{ClementeParysSalvatiWalukiewicz:Diagonal:LICS16}
(later extended to safe schemes over finite trees~\cite{BarozziniClementeColcombetParys:ICALP:2020}),
and for classes $\D$ more restrictive than the regular languages,
such as in~\cite{Place-dotdepth,PlaceZeitoun:TOCL:2019}.
The separability and membership problems have also been studied for several classes of infinite-state systems,
such as vector addition systems~\cite{ClementeCzerwinskiLasotaPaperman:STACS:2017,ClementeCzerwinskiLasotaPaperman:Parikh:ICALP:2017,CzerwinskiZetzsche:LICS:2020},
well-structured transition systems~\cite{CzerwinskiLasotaMeyerMuskallaKumarSaivasan:CONCUR:2018},
one-counter automata~\cite{CzerwinskiLasota:LMCS:2019},
and timed automata~\cite{ClementeLasotaPiorkowski:ICALP:2020,ClementeLasotaPiorkowski:CONCUR:2020}.
Recent developments on efficient algorithms solving parity games
are based on the ability to find a simple separator,
yielding both upper bounds on the problem,
and lower bounds for a wide family of algorithms~\cite[Chapter~3]{BojanczykCzerwinski:Toolbox:2018,CzerwinskiDaviaudFijalkowJurzinskiLazicParys:SODA:2019}.
Finally, it is worth mentioning that games have already been successfully used to provide several characterisation results,
such as in~\cite{ColcombetLoding:ICALP:2008,ColcombetKuperbergLodingVandenBoom:CSL:2013,colcombet_hab,bojanczyk_star_games,SkrzypczakWalukiewicz:ICALP:2016,CavallariMichalewskiSkrzypczak:MFCS:2017}.

\subparagraph{Outline.}
In \cref{sec:prelim} we introduce automata and other mathematical preliminaries.
In \cref{sec:det_ind,sec:det_gen,sec:game_gen,sec:game_ind} we present the game-theoretic characterisations
of the separability problems we consider.
We believe this is the most interesting aspect of this work.
%
%
\Cref{ap:complexity} is devoted to an~analysis of the computational complexity of our decision methods
leading to the proof of the announced \cref{thm:complexity}.

\section{Preliminaries}
\label{sec:prelim}

A~nonempty finite set $\Sigma$ of \emph{letters} $a\in\Sigma$ is called an~\emph{alphabet}.
A~(\emph{$\Sigma$-labelled}) \emph{tree} is a function
$t \colon \set{\L, \R}^* \to \Sigma$
assigning to each \emph{node} $u \in \set{\L, \R}^*$
a~\emph{label} $t(a) \in \Sigma$. The \emph{root} of a~tree 
 is denoted~$\epsilon$.
The set of all $\Sigma$-labelled trees is denoted $\trees_\Sigma$.
The symbols $\L$, $\R$ are called \emph{directions}
and a~\emph{branch} is an~infinite sequence thereof $d_0d_1\cdots\in\set{\L,\R}^\omega$.
A~tree $t$ is uniquely defined by the set of its \emph{paths}
$\paths t = \setof{\tuple{a_0,d_0} \tuple {a_1,d_1} \cdots \in (\Sigma \times \set{\L, \R})^\omega}
{\forall i.\ a_i = t(d_0 d_1 \cdots d_{i-1})}$,
which is extended to languages pointwise as
$\paths L = \setof{\paths t}{t \in L}$. 

\subsection{Automata}

Fix a nonempty finite set of \emph{priorities} $C \subseteq \N$.
A~\emph{(top-down, nondeterministic, parity, tree) automaton} is a~tuple
$
    \A = \tuple{ \Sigma, Q, q_0, \Omega, \Delta},
$
where $\Sigma$ is a~finite alphabet,
$Q$ is a~finite set of \emph{states},
amongst which $q_0 \in Q$ is the~\emph{initial state},
$\Omega\colon Q \to C$ assigns a~priority to every state,
and $\Delta \subseteq Q \times \Sigma \times Q \times Q$ is a~set of \emph{transitions}.
%
The priority function $\Omega$ is extended to a~transition $\delta = \tuple{q, \_, \_, \_}$
as $\Omega(\delta) := \Omega(q)$, pointwise to an~infinite sequence of states
$\Omega(q_0q_1 \cdots) := \Omega(q_0)\Omega(q_1)\cdots \in C^\omega$
and transitions $\Omega(\delta_0 \delta_1 \cdots) = \Omega(\delta_0)\Omega(\delta_1) \cdots \in C^\omega$.
An~infinite sequence of priorities $c_0c_1 \cdots \in C^\omega$ is \emph{accepting}
if the maximal priority occurring infinitely often is even.
Similarly, an~infinite sequence of states $\rho = q_0q_1 \cdots \in Q^\omega$
or of transitions $\rho = \delta_0\delta_1\cdots \in \Delta^\omega$
is accepting whenever $\Omega(\rho)$ is accepting.
We write $\Delta(q, a) = \set{\tuple{q, a, q_\L, q_\R} \in \Delta}$
for the set of transitions from a~state $q\in Q$ over a~letter $a\in\Sigma$,
and $\Delta(a) = \bigcup \setof{\Delta(q, a)}{q \in Q}$ for all transitions over~$a$.
We extend the notation above to an~infinite path
$b = \tuple{a_0,d_0} \tuple{a_1,d_1} \cdots \in (\Sigma \times \set{\L, \R})^\omega$
by writing $\Delta(b)$ for the set of infinite sequences of transitions
$\vec\delta = \delta_0 \delta_1 \cdots \in \Delta^\omega$
of the form $\delta_i = (q_i, a_i, q_{\L,i}, q_{\R,i})$ for every $i$,
which are \emph{conform to $b$}
in the sense that $q_0$ is the initial state of the automaton and $q_{i+1} = q_{d_i,i}$. 

A~\emph{run} of an~automaton $\A$ as above over a tree $t\in \trees_\Sigma$ is a~$Q$-labelled tree $\rho\in\trees_Q$ s.t.~$\rho(\epsilon)=q_0$ is the initial state and for every node in the tree $u\in\set{\L,\R}^*$ the quadruple $\big(\rho(u),t(u),\rho(u\L),\rho(u\R)\big)$ belongs to $\Delta$. Such a~run is \emph{accepting} if for every branch $d_0d_1\cdots\in\set{\L,\R}^\omega$ the sequence of states $\big(\rho(d_0\cdots d_{i-1})\big)_{i\in\omega}$ is accepting. 
The set of all trees $t\in\trees_\Sigma$ s.t.~$\A$ has an~accepting run over $t$ is denoted $\lang\A$ and is called the \emph{language} recognised by $\A$. The corresponding \emph{path language} 
is $\blang \A := \paths {\lang \A} \subseteq (\Sigma \times \set{\L, \R})^\omega$. If $q\in Q$ is a~state of an~automaton $\A$ then by $\A_q$ we denote the same automaton as $\A$ but with the initial state $q_0$ changed to $q$. Thus, $\lang {\A_q}$ is the set of trees over which $\A$ has an~accepting run $\rho$ starting at $\rho(\epsilon)=q$.
In the rest of the paper we assume that all states $q$ in an automaton are \emph{productive}
in the sense that $\lang {\A_q} \neq \emptyset$.

\subsection{Deterministic and game automata}

We say that $\A$ is a \emph{game automaton} if, for every $q \in Q$ and $a \in \Sigma$,
either we have a \emph{conjunctive transition} $\Delta(q, a) = \set{\tuple{q, a, q_\L, q_\R}}$ 
or two \emph{disjunctive transitions} $\Delta(q, a) = \set{\tuple{q, a, q_\L, \top}, \tuple{q, a, \top, q_\R}}$
(c.f.~\cite[Definition~3.2]{FacchiniMurlakSkrzypczak:TOCL:2016}),
where $\top\neq q_0$ represents a~distinguished state in $Q$ accepting every tree (i.e., $\lang {\A_\top}=\trees_\Sigma$)
and $q_\L, q_\R \neq \top$.
%
%
An automaton $\A$ is \emph{deterministic} if it is a game automaton with only conjunctive transitions
and in this case for every tree $t\in\trees_\Sigma$ there exists a~unique run $\rho$ of $\A$ over $t$.
%
A tree language $L$ is \emph{deterministic}, resp., \emph{game},
if it can be recognised by some deterministic, resp., game automaton.
%
%
%
%
%
Game automata can be complemented with very low complexity
by just increasing every priority by one
and by swapping conjunctive and disjunctive transitions.
\begin{restatable}{lemma}{remGameAutomataDual}
    \label{rem:game-automata-dual}
    If $\A$ is a~game parity tree automaton,
    then $ \trees_\Sigma \setminus \lang \A$ can be recognised by a~game parity tree automaton
    with the same number of states and priorities.
\end{restatable}
%

\begin{proof}
    Let $\A = \tuple{ \Sigma, Q, q_0, \Omega, \Delta}$ be a~game automaton.
    Its complement is the game automaton $\A^\mathrm c$ obtained by swapping conjunctive transitions with disjunctive ones, and vice versa, and by increasing priorities by one.
    Formally, $\A^\mathrm c = \tuple{\Sigma, Q, q_0, \Omega^\mathrm c, \Delta^\mathrm c}$
    where $\Omega^\mathrm c(q) = \Omega(q) + 1$
    and the set of transitions $\Delta^\mathrm c$ is obtained by dualising $\Delta$ as follows:
    for every $q \in Q$ and $a \in \Sigma$,
    if $\Delta(q, a) = \set{\tuple{q, a, q_\L, q_\R}}$ is conjunctive
    then $\Delta^\mathrm c(q, a) = \set{\tuple{q, a, q_\L, \top}, \tuple{q, a, \top, q_\R}}$ is disjunctive, 
    and symmetrically in the other case.
    It is standard to check that $\lang {\A^\mathrm c}=\trees_\Sigma\setminus \lang \A$.
\end{proof}

\subsection{Determinisation over \texorpdfstring{$\omega$-words}{w-words}}

A \emph{nondeterministic $\omega$-word parity automaton}
is a tuple $\A = \tuple{\Sigma, Q, q_0, \Omega, \Delta}$
where $\Sigma$ is a finite input alphabet,
$Q$ is a finite set of states,
$q_0 \in Q$ is an initial state,
$\Omega\colon Q \to C$ assigns to each state a priority in $C$,
and $\Delta \subseteq Q \times \Sigma \times Q$ is a transition relation.
The notions of runs and the accepted language $\lang \A \subseteq \Sigma^\omega$ are standard \cite{AutomataLogicsInfiniteGames:2002}.
We recall that nondeterministic $\omega$-word parity automata can be determinised
with an~exponential complexity in the number of states and a~polynomial complexity in the number of priorities.
We will use this fact in later proofs.

\begin{restatable}[\protect{c.f.~\cite{ScheweVarghese:MFCS:2014}}]{lemma}{lemNPADPA}
    \label{lem:NPA2DPA}
    A nondeterministic $\omega$-word parity automaton $\A$ with $n$ states and~$k$ priorities
    can be converted to an equivalent deterministic parity automaton
    with $n' = 2 \cdot (n \cdot (k+1))^{n \cdot (k+1)} \cdot (n \cdot (k+1))!$ states and $k' = 2 \cdot n \cdot (k+1)$ priorities.
\end{restatable}

In order to prove \cref{lem:NPA2DPA},
we first prove the following result allowing us to convert nondeterministic parity to nondeterministic Büchi automata.

\begin{lemma}
    \label{lem:NPA2NBA}
    A nondeterministic $\omega$-word parity automaton $\A$ with $n$ states and $k$ priorities
    can be converted to an equivalent nondeterministic $\omega$-word Büchi automaton with $n \cdot (k+1)$ states.
\end{lemma}

\begin{proof}
    Let $\A$ have $n = \card P$ states and $k = \card C$ priorities.
    Automaton $\B$ has states of the form $Q = P \cup P \times C$.
    In the first phase $\B$ just simulates $\A$,
    until it goes to a state of the form $\tuple{p, c}$
    by nondeterministically guessing an even priority $c \in C$
    and checking that $c$ is visited infinitely often
    and no larger priority is visited in the rest of the run.
\end{proof}

\cref{lem:NPA2DPA} follows from \cref{lem:NPA2NBA}
and the following result allowing us to convert from nondeterministic Büchi to deterministic parity automata.

\begin{lemma}[\protect{\cite[Theorem~3.10]{Piterman:LMCS:2007}}]
    \label{lem:NBA2DPA}
    A nondeterministic $\omega$-word Büchi automaton $\A$ with $n$ states
    can be converted to an equivalent deterministic $\omega$-word parity automaton $\B$
    with $2 \cdot n^n \cdot n!$ states and $2 \cdot n$ priorities.
\end{lemma}

\subsection{Games}
\label{ap:games}

In this section we formalise the framework of games used throughout the paper.
These are variants of two-player zero-sum perfect information games on graphs of infinite duration
where some intermediate positions are hidden.
The default names of the two players are \PI and \PII, however in most games it will be more convenient to work with some more meaningful names. If $P \in \set{\PI, \PII}$ is a~player then the other player is called the \emph{opponent} of~$P$.
To specify a~game we need to define an~arena and a~winning condition. 
An~\emph{arena} of a~game consists of: a~nonempty set $V$ of \emph{positions}, an~\emph{initial position} $v_0\in V$, a~\emph{position update function} $\eta$, a~finite sequence $\big( (P^{(0)},X^{(0)}),\ldots, (P^{(n)},X^{(n)})\big)$ of possible \emph{decisions} that players can make during a~round, and \emph{restrictions} $O_v$, one for each position $v\in V$.
Each \emph{decision} $\tuple{P^{(k)},X^{(k)}}$ is left in the hands of one of the players $P^{(k)} \in \set{\PI, \PII}$
and is taken from some fixed nonempty finite set $X^{(k)}$ of \emph{possible choices}. The product of all the sets of possible choices $O:=X^{(0)}_v\times\ldots\times X^{(n)}_v$ is called the set of \emph{round outcomes}. The \emph{position update function} is of the type $\eta\colon V\times O\to V$. An~arena can additionally \emph{restrict} some decisions of some players depending on the current position $v\in V$ of the game and some previous decisions in this round using the \emph{restrictions}, i.e.,~nonempty subsets $O_v\subseteq O$, indexed by the positions $v\in V$.

At the $i$-th round starting in a~position $v_i\in V$ players declare their choices $x^{(k)}_i\in X^{(k)}$ in the order specified by the arena and according to the imposed restrictions. More formally, they inductively define a~vector $o_i=\tuple{x^{(0)}_i,\ldots,x^{(n)}_i}$, keeping the invariant that
\begin{equation}
\label{eq:game-invariant}
\tuple{x^{(0)}_i,\ldots,x^{(k-1)}_i}\in \proje_{0,\ldots,k-1} (O_v),
\end{equation}
i.e.,~the constructed vector belongs to the projection of $O_v$ onto the coordinates $0,\ldots,k{-}1$. This guarantees that the successive player $P^{(k)}$ has always at least one choice $x^{(k)}\in X^{(k)}$ satisfying the invariant. Once the whole vector $o_i=\tuple{x^{(0)}_i,\ldots,x^{(n)}_i}$ is constructed, the $i$-th round is finished. The next position of the game is defined by $v_{i+1}:=\eta(v_i,o_i)$. A~play of the game is the sequence of round outcomes $o_0o_1\cdots\in O^\omega$ (the visited positions are implicit, however can easily be computed using $\eta$).

A~\emph{winning condition} of a~game specifies which infinite plays $o_0o_1\cdots\in O^\omega$ are considered winning for one of the players $P$; with the remaining plays losing for $P$ and winning for the opponent of $P$. Formally, a~winning condition is just a~subset of $O^\omega$.

Some of the considered games are \emph{positionless}, i.e.,~there is only a~single position $V=\{v_0\}$. An~arena is called \emph{finite} if $V$ is finite.

A~\emph{strategy} of a~player $P$ for a~game is a~tuple $\M=\tuple {M, \ell_0, \overline{v}, \tau}$ where:
$M$ is a~set of \emph{memory states},
$\ell_0\in M$ is an~\emph{initial memory state},
$\overline{o}=(\overline{x},\ldots,\overline{z})$ is a~vector of \emph{decision functions}, one function for each decision of the player $P$,
and $\tau$ is a~\emph{memory update function} that maps a~position of the game $v\in V$, a~memory state $\ell\in M$, and a~vector of choices of the opponent $(x',\ldots,z')$ into the next memory value $\ell'=\tau(v,\ell,x',\ldots,z')\in M$.
The domain of a~decision function $\overline{y}$ for a~decision $(P,X^{(k)})$ allowing the player $P$ to choose $x^{(k)}\in X^{(k)}$ is the~product of $V$, $M$, and all the possible previous sets of possible options of the opponent in a~round. The range of $\overline{y}$ is $Y$, with the restriction that the player $P$ needs to preserve the invariant as in~\eqref{eq:game-invariant}.

A~strategy $\sigma$ is of \emph{finite memory} if $M$ is a~finite set. Notice that each finite memory strategy for a~finite game is a~finite object that can be effectively represented. A~strategy is \emph{positional} if $M=\{\ell_0\}$. In the case of a~positional strategy there is a~unique memory value, the function $\tau$ is trivial, and we ignore the $M$ argument of the decision functions. Similarly, if the given arena is positionless then we ignore the $V$ argument of the functions above.

Fix a~strategy $\M$ of a~player $P$. Such a~strategy determines the way in which the player $P$ should make her choices. Consider the $i$-th round of the game, starting in a~position $v_i\in V$ and with a~memory state $\ell_i\in M$ (for $i=0$ the memory state $\ell_0$ is the initial memory value). The consecutive choices of $P$ in this round are given by the decision functions in $\overline{v}$ applied to $v$, $\ell_i$, and the previous choices of the opponent. Once the round is finished with opponent's choices $\tuple{x'_i,\ldots,z'_i}$, we take $\ell_{i+1}:=\tau(v_i,\ell_i,x'_i,\ldots,z'_i)$.

A~play $o_0o_1\cdots\in O^\omega$ that is obtained according to the policy above is said to be \emph{conform} to the strategy $\M$. Notice that if $\M$ and $\M'$ are two strategies of the two players then there exists a~unique play $o_0o_1\cdots\in O^\omega$ that is conform to both of them---this play can be defined inductively in the standard way.

A~strategy $\M$ of a~player $P$ is said to be \emph{winning} if all the plays conform to $\M$ are winning for $P$. Because of the observation above, at most one of the players has a~winning strategy. We say that a~position $v\in V$ is \emph{winning} for $P$ if $P$ has a~winning strategy in the game with the initial position $v_0$ set to $v$. We say that a~game is \emph{determined} if exactly one player has a~winning strategy.

\paragraph{Games on graphs.}

To formally prove the results of finite-memory determinacy of the games involved in our work, we show how to reduce them to the standard framework of \emph{games on graphs}.

An arena of a~game on graph is specified by a~directed graph $\tuple{V, E}$ without dead-ends (the elements of $V$ are called \emph{positions}),
an~\emph{initial position} $v_0\in V$,
an~\emph{ownership partition} $V=V_{\PI}\cup V_{\PII}$ into two disjoint sets,
and a~labelling function $\rho\colon E\to\Sigma$. A~\emph{play} of such a~game is constructed inductively by the players, starting from the initial position $v_0$. At the $i$-th round, with the current position $v_i\in V_P$, the player $P$ chooses an~edge $\tuple{v_i,v_{i+1}}\in E$, defining the consecutive position $v_{i+1}$.

A~winning condition of $\PI$ in such a game is a language $W\subseteq \Sigma^\omega$. A~play as above is winning for $\PI$ if the $\omega$-word $\rho(v_0,v_1)\rho(v_1,v_2)\cdots$ belongs to $L$.

We will now show how to reduce a~game defined according to our definition into a game on graph. Consider a~game with an arena consisting of a~set of positions $V$; an initial position $v_0\in V$; a~position update function $\eta$; decisions $\big( (P^{(0)},X^{(0)}),\ldots, (P^{(n-1)},X^{(n-1)})\big)$; and restrictions $(O_v)_{v\in V}$. Let $O$ be the set of round outcomes. Consider a~graph with the set of positions
\[V':=\bigcup_{v\in V} \left(\{v\}\times \bigcup_{k=0,\ldots,n} \proje_{0,\ldots,k-1}(O_v)\right),\]
where $\proje_{0,\ldots,-1}(O_v)$ is the singleton $\{()\}$ consisting of the empty tuple $()$. The initial position is $v'_0:=\tuple{v_0,()}$.

Let the set of edges $E':=E'_0\cup E'_1$ consists of the following two types of edges. The first is defined as
\[E'_0 := \left(\big(v,(x^{(0)},\ldots,x^{(k-1)})\big), \big(v,(x^{(0)},\ldots,x^{(k-1)},x^{(k)})\big)\right),\]
where $c\in V$ and $(x^{(0)},\ldots,x^{(k-1)},x^{(k)})\in\proje_{0,\ldots,k}(O_v)$. The second is defined as
\[E'_1 := \left(\big(v,o\big), \big(v',()\big)\right),\]
where $v\in V$, $o\in O_v$, and $\eta(v,o)=v'$. Let $\Sigma=O\cup\{\epsilon\}$ and the labelling $\rho$ be defined as $\rho(e):=\epsilon$ for $e\in E'_0$; and $\rho\left(\big(v,o\big), \big(v',()\big)\right):=o$ for $\left(\big(v,o\big), \big(v',()\big)\right)\in E'_1$.

Given a~winning condition $W\subseteq O^\omega$, we define the new winning condition $W'\subseteq\Sigma^\omega$ by skipping the symbols $\epsilon$ (notice that the shape of the arena ensures that every $n$-th edge is labelled by an element of $O$).

\begin{claim}
\label{cl:translation}
Each strategy (understood in the standard sense) of a~player $P$ in the corresponding game on graph can be translated into a~strategy of the shape $\M=\tuple {M, \ell_0, \overline{v}, \tau}$ in the original game. Moreover, this translation preserves the size of the memory and maps a~winning strategy into a~winning strategy.
\end{claim}

\subparagraph{Determinacy.}
\label{app:determinacy}

We rely on two important known results of determinacy of the considered games, obtained directly from the known results via \cref{cl:translation}.

\begin{theorem}[\protect{\cite[Theorem~$1'$]{BuchiLandweber:AMS:1969}}]
    \label{thm:BuchiLandweber}
	Consider a~game arena with the set of round outcomes $O$. Assume that the set of winning plays of \PI is an~$\omega$-regular language over the alphabet $O$. Then one of the players has a~finite memory winning strategy in this game. Moreover, such a~strategy can be effectively computed based on a~representation of a~finite arena and the winning condition.
\end{theorem}

We say that a~winning condition $W\subseteq O^\omega$ is a~\emph{Rabin} condition over $O$ if
\[W=\big(\mathrm{inf}(E_0)\cap \mathrm{fin}(F_0)\big)\cup \cdots \cup \big(\mathrm{inf}(E_n)\cap \mathrm{fin}(F_n)\big)\]
s.t.~for $k=0,\ldots,n$ we have $E_k,F_k\subseteq O$ and
\begin{align*}
\mathrm{inf}(E_k) &:= \big\{(o_i)_{i\in\omega}\in O^\omega\sep \text{$o_i\in E_k$ for infinitely many $i\in\omega$}\big\}\\
\mathrm{fin}(F_k) &:= \big\{(o_i)_{i\in\omega}\in O^\omega\sep \text{$o_i\in F_k$ for only finitely many $i\in\omega$}\big\}.
\end{align*}
Notice that the family of Rabin conditions is closed under union.

If $C\subseteq\N$ is a~finite set of priorities then the set of parity accepting sequences $(c_i)_{i\in\omega}\subseteq C^\omega$ can be written as a~Rabin condition with $E_k=\{n\in C\sep n \geq 2k\}$ and $F_k=\{n\in C\sep n \geq 2n{+}1\}$ for $0\leq k\leq\frac{\max C}{2}$. Therefore, the parity condition is a~special case of Rabin condition. Similarly, the complement of a~parity condition is also a~Rabin condition.

\begin{theorem}[\protect{\cite[Lemma~9]{Klarlund:APAL:1994}; c.f.~also \cite[Theorem 4]{Graedel:STACS:2004} and \cite[Theorem 7.12]{Kopczynski:PhD:2008}}]
\label{thm:rabin-positional}
	Consider~a~game~arena~with~the~set~of~round~outcomes~$O$. Assume~that~the~set~of winning~plays~of \PI is~a~Rabin~condition~over $O$. Then \PI has~a~\emph{uniform}~positional~strategy $\M$ in~that~game, i.e.,~a~positional~strategy~such~that~for~every~position $v\in V$,~if~$v$~is~winning~for \PI then~$\M$~is~a~winning~strategy~from~$v$.
\end{theorem}

\subparagraph{Complexity.}

Finally, we recall that games on graphs with parity winning conditions can be solved in quasi-polynomial time.

\begin{lemma}[\protect{\cite[Theorem 2.9]{CaludeJainKhoussainovLiStephan:STOC:2017}}]
    \label{lem:parity:games}
    A parity game with $n$ positions and $k = \card C$ priorities
    can be solved in deterministic time $O(n^{\log k + 6})$.
\end{lemma}
This will be used in our complexity analysis in \cref{ap:complexity}.
In fact, already a na\"ive parity game algorithm with complexity $O(n^k)$ (i.e., polynomial in $n$ and exponential in $k$)
would suffice for our purposes since we will instantiate it on game graphs of exponential size and polynomially many priorities.

\subsection{Acceptance games}

We present a game-theoretic view on accepting runs for automata
based on the framework from \cref{ap:games}.
This will serve both as an example of the kind of games that we consider throughout paper,
and as a technical tool in the proofs from \cref{sec:game_gen,sec:game_ind}.

Let $t \in \trees_\Sigma$ be a tree.
The \emph{acceptance game} $\AcceptanceGame \A t$ is played in rounds by two players, \Automaton and \Pathfinder.
The goal of \Automaton is to show that $t \in \lang \A$;
 \Pathfinder has the complementary objective $t \not \in \lang \A$.
%

\begin{gamebox}{Acceptance game $\AcceptanceGame \A t$}
    At the $i$-th round starting at a~\emph{position} $v_i = \tuple{u_i, q_i} \in V := \set{\L,\R}^* \times Q$:
    \begin{gameize}
        \item[\Choice{A}{\delta}] \Automaton plays a transition $\delta_i = \tuple{q_i, t(u_i), q_{\L,i}, q_{\R,i}} \in \Delta(q_i,t(u_i))$.
        \item[\Choice{P}{d}] \Pathfinder plays a direction $d_i \in \set{\L, \R}$.
    \end{gameize}
	The next position is $v_{i+1} := (u_i d_i, q_{d_i, i})$.
\end{gamebox}
The initial position is $v_0:=\tuple{\epsilon,q_0}$.
\Automaton \emph{wins} the resulting infinite play $\pi = \tuple{\delta_0,d_0} \tuple{\delta_1,d_1} \cdots$
if the sequence of transitions $\delta_0 \delta_1 \cdots$ is accepting.

The following proposition is folklore.

\begin{proposition}
    Let $t\in\trees_\Sigma$ and $\A$ be an~automaton over the alphabet $\Sigma$.
   \Automaton wins the acceptance game $\AcceptanceGame \A t$ if, and only if,
    $t\in\lang\A$.
\end{proposition}

\subsection{Disjointness games}
\label{ssec:disjointness-game}

Let $\A$ and $\B$ be two nondeterministic automata.
We recall a~standard game used to characterise whether $\lang \A \disjoint \lang \B$.
This will be crucial in the correctness proofs throughout \cref{sec:det_ind,sec:det_gen,sec:game_gen,sec:game_ind}.
The \emph{disjointness game} $\DisjointnessGame \A \B$ is played by two players,
\Automaton and \Pathfinder. \Automaton's aim is to incrementally build a~tree accepted by both $\A$ and $\B$,
witnessing $\lang \A\cap \lang \B \neq \emptyset$,
while \Pathfinder has the opposite objective.%
\footnote{The disjointness game could equivalently be phrased as a nonemptiness game
for the product automaton $\A\times\B$ recognising $\lang \A\cap \lang \B$.
However, in our technical development it will be more direct to use the disjointness game.}
The set of positions of the game is $Q^\A\times Q^\B$, and the initial position is $(q^\A_0,q^\B_0)$.

\begin{gamebox}{Disjointness game $\DisjointnessGame \A \B$}
    At the $i$-th round starting at a~position $(q^\A_i, q^\B_i)$:
    \begin{gameize}
        \item[\Choice{A}{a}] \Automaton plays a~letter $a_i\in\Sigma$.
		\item[\Choice{A}{\delta^\A}] \Automaton plays a~transition $\delta^\A_i = \tuple{q^\A_i, a_i, q^\A_{\L,i}, q^\A_{\R,i}} \in \Delta^\A(q^\A_i, a_i)$.
		\item[\Choice{A}{\delta^\B}] \Automaton plays a~transition $\delta^\B_i = \tuple{q^\B_i, a_i, q^\B_{\L,i}, q^\B_{\R,i}} \in \Delta^\B(q^\B_i, a_i)$.
        \item[\Choice{P}{d}] \Pathfinder plays a~direction $d_i \in \set{\L, \R}$.
    \end{gameize}
	The next position is $(q^\A_{d_i,i}, q^\B_{d_i,i})$.
\end{gamebox}
Let the resulting infinite play be $\pi = \tuple{a_0,\delta^\A_0,\delta^\B_0,d_0} \tuple{a_1,\delta^\A_1,\delta^\B_1,d_1} \cdots$. Such a~play induces an~infinite path $b = \tuple{a_0,d_0} \tuple{a_1,d_1} \cdots$ and two sequences of transitions $\vec\delta^\A := \delta^\A_0\delta^\A_1 \cdots$ and $\vec\delta^\B := \delta^\B_0\delta^\B_1 \cdots$.
The rules of the game guarantee that $\vec\delta^\A \in \Delta^\A(b)$ and $\vec\delta^\B \in \Delta^\B(b)$.
\Automaton wins the play $\pi$ if both sequences $\vec\delta^\A$ and $\vec\delta^\B$ are accepting.

In the rest of the paper it will be more useful to consider \Pathfinder's point of view.
Since her winning condition can be presented as Rabin condition (see \cref{app:determinacy}),
whenever she wins, she has a memoryless (i.e.,~$M=\{\ell_0\}$) winning strategy.
Such a~memoryless strategy for \Pathfinder in the disjointness game
can be represented by a function $\PP\colon \left(\bigcup_{a\in\Sigma} \Delta^\A(a)\times \Delta^\B(a)\right)\to \set{\L,\R}$,
which we call a \emph{pathfinder}.
%

\begin{restatable}{lemma}{lemDisjointPathfinder}
    \label{cor:disjoint-pathfinder}
    If $\lang \A\disjoint \lang \B$ then there is a~pathfinder $\PP$
    which is winning for \Pathfinder in the disjointness game $\DisjointnessGame \A \B$.
\end{restatable}

\begin{proof}
The winning condition for \Pathfinder is a~disjunction of two properties:
Either $\vec\delta^\A$ or $\vec\delta^\B$ is rejecting.
Both these properties are complements of parity conditions.
Therefore, the winning condition of \Pathfinder is a~Rabin condition.
Since memoryless winning strategies suffice for games with Rabin winning conditions,
it follows that if \Pathfinder wins then she has a memoryless winning strategy.

In general, such a~positional strategy is of the form $\tuple{ \{\ell_0\}, \ell_0, \overline{d}, \tau}$, with $\tau$ constantly equal $\ell_0$ and $\overline{d}\colon (q^\A, q^\B, a, \delta^\A, \delta^\B)\mapsto d$, for $q^\A\in Q^\A$, $q^\B\in Q^\B$, $a\in\Sigma$, $\delta^\A \in \Delta^\A(q^\A, a)$, $\delta^\B \in \Delta^\B(q^\B, a)$, and $d\in \set{\L,\R}$. Due to the redundancy within the arguments of the decision function $\overline{d}$, we can represent such a~strategy by a~function $\PP\colon \left(\bigcup_{a\in\Sigma} \Delta^\A(a)\times \Delta^\B(a)\right)\to \set{\L,\R}$, i.e.,~a~pathfinder.
\end{proof}

\cref{rem:pathfinder-property} below follows directly from the construction of $\PP$
and the fact that the strategy $\M$ used to obtain it is winning.

\begin{restatable}{corollary}{remPathfinderProperty}
\label{rem:pathfinder-property}
Assume that $\lang \A\disjoint \lang \B$ and let $\PP$ be a~pathfinder as above. Let $b=\tuple{a_0,d_0}\tuple{a_1,d_1}\cdots\in (\Sigma\times\set{\L,\R})^\omega$ be a~path and $\vec\delta^\A=\delta^\A_0\delta^\A_1\cdots\in\Delta^\A(b)$, $\vec\delta^\B=\delta^\B_0\delta^\B_1\cdots\in\Delta^\B(b)$ be two sequences of transitions of these automata that are conform to~$b$. If for every $i\in\omega$ we have $\PP(\delta^\A_i,\delta^B_i)=d_i$ then at least one of the sequences $\vec\delta^\A$ and $\vec\delta^\B$ is rejecting.
\end{restatable}

%
The construction from \cref{cor:disjoint-pathfinder} above has a~specific property
when one of the involved automata (e.g.,~$\A$) is a~game automaton.
Since we assume that every state is productive,
positions of the form $(\top, q^\B)$ are losing for \Pathfinder in $\DisjointnessGame \A \B$.
Therefore, without loss of generality we can assume that the pathfinder $\PP$ satisfies the following observation.

\begin{remark}
\label{rem:pathfinder-for-game}
 Consider a transition $\delta^\A=\tuple{q^\A,a,q^\A_\L,\top}$ (resp., $\delta^\A=\tuple{q^\A,a,\top,\allowbreak q^\A_\R}$) in a game automaton $\A$.
 Then, $\PP(\delta^\A,\_)$ is constantly equal to $\L$ (resp., $\R$).
\end{remark}

\subsection{Deterministic separability over \texorpdfstring{$\omega$-words}{w-words}}
\label{ap:det_ind}

In this section, we provide full details for the separability problem for $\omega$-words sketched in the introduction.
This will also serve as an introduction to the more challenging separability problems over infinite trees considered in the rest of the paper.
Let $\A$, $\B$ be two nondeterministic parity automata over $\omega$-words and let $C\subseteq\N$ be a~set of priorities.
Consider the following \emph{$C$-deterministic-separability game} $\OmegaSeparabilityGame \A \B C$.

\begin{gamebox}{$C$-deterministic-separability game over $\omega$-words $\OmegaSeparabilityGame \A \B C$}
    At the $i$-th round:
    \begin{gameize}
        \item[\Choice{S}{c}] \Separator plays a priority $c_i \in C$.
        \item[\Choice{I}{a}] \Input plays a letter $a_i \in \Sigma$.
    \end{gameize}
\end{gamebox}
\noindent
\Separator wins an infinite play
$\pi = \tuple{c_0,a_0} \tuple{c_1,a_1} \cdots \in (C \times \Sigma)^\omega$
if the following two conditions are both satisfied (i.e.,~if $\pi\in W := \Win\A \cap \Win\B \subseteq (C \times \Sigma)^\omega$):
\begin{winnize}
    \item~$\pi\in\Win\A$: If $a_0 a_1 \cdots \in \lang \A$, then $c_0c_1 \cdots$ is accepting.
    \item~$\pi\in\Win\B$: If $a_0 a_1 \cdots \in \lang \B$, then $c_0c_1 \cdots$ is rejecting.
\end{winnize}

\begin{restatable}{lemma}{lemWordsDetSep}
    \Separator wins $\OmegaSeparabilityGame \A \B C$
    if, and only if, $\lang \A$, $\lang \B$ can be separated by a~deterministic parity automaton with priorities in $C$.
\end{restatable}


A~strategy for \Separator in $\OmegaSeparabilityGame \A \B C$
is of the form $\M = \tuple{ M, \ell_0, \overline{c}\colon M\to C, \tau\colon M\times\Sigma\to M}$.
The rest of this section is devoted to the~proof of this lemma.

\begin{proof}[Soundness]
Assume that \Separator wins $\OmegaSeparabilityGame \A \B C$, let $\M$ be his finite-memory winning strategy as above.
Consider a~candidate separating automaton
$\S := \tuple{ \Sigma, M, \ell_0, \Omega, \Delta}$
having the same set of states $M$ as $\M$'s memory states, and the same initial state $\ell_0$,
where state $\ell$'s priority $\Omega(\ell) := \overline{c}(\ell)$ is provided directly by the decision function $\overline{c}$,
and the set of transitions is defined according to the memory update function $\tau$ as
\begin{align*}
\Delta = \big\{\tuple{\ell, a, \tau(\ell, a)}\mid \ell \in M, a \in \Sigma\big\}.
\end{align*}
Clearly $\S$ is a~$C$-deterministic automaton over $\omega$-words. Moreover, $\Win\A$ guarantees that $\lang\A\subseteq\lang\S$, while $\Win\B$ guarantees that $\lang\B\disjoint \lang\S$. Therefore, $\S$ is the required separator.
\end{proof}

\begin{proof}[Completeness]
Let $\S = \tuple{ \Sigma, Q, q_0, \Omega, \Delta}$ be a~separating automaton.
We build a finite-state winning strategy for \Separator
$\M = \tuple{ Q, q_0, \overline{c}, \tau}$ over the same set of states $Q$
s.t.~$\overline{c}(q) = \Omega(q)$ is $q$'s priority in $\S$,
and $\tau(q, a) = q'$ for the unique $q'$ s.t.~$\tuple{q, a, q'} \in \Delta$.
It is immediate to show that $\M$ is winning from the fact that $\S$ is a~separator.
\end{proof}

\section{Separability by deterministic automata with priorities in \texorpdfstring{$C$}{C}}
\label{sec:det_ind}

In this section we present a game-theoretic characterisation of separability over infinite trees by deterministic automata with a fixed finite set of priorities $C \subseteq \mathbb N$.
Let $\A$, $\B$ be two nondeterministic automata over infinite trees.
We extend the game over $\omega$-words from the introduction (and formally defined in \cref{ap:det_ind})
with two additional actions:
a \emph{selector} for \Separator and a~direction for \Input.

\begin{gamebox}{$C$-deterministic-separability game $\CDeterministicSeparabilityGame \A \B C$}
    At the $i$-th round:
    \begin{gameize}
        \item[\Choice{S}{c}] \Separator plays a priority $c_i \in C$.
        \item[\Choice{I}{a}] \Input plays a letter $a_i \in \Sigma$.
        \item[\Choice{S}{f}] \Separator plays a \emph{selector} $f_i \in \set{\L, \R}^{\Delta^\B(a_i)}$.
        \item[\Choice{I}{d}] \Input plays a direction $d_i \in \set{\L, \R}$.
    \end{gameize}
\end{gamebox}

\noindent
Intuitively, a~selector encodes a direction for each (relevant) transition of $\B$
and this is used for the correctness of the separator.
(In \cref{app:naive-det-gen} we consider a simpler variant without selectors
and we discuss which separability problem it captures.)
Let the resulting infinite play be $\pi = \tuple{c_0, a_0, f_0, d_0} \tuple{c_1, a_1, f_1, d_1} \cdots$,
with the induced infinite path $b := \tuple{a_0, d_0} \tuple{a_1, d_1} \cdots$.
\Separator wins the play $\pi$
if the following two conditions are satisfied:
\begin{winnize}
    \item $\pi\in\Win\A$: If there exists an~accepting sequence of transitions
    $\vec\delta^\A = \delta_0^\A \delta_1^\A \cdots \in \Delta^\A(b)$, then $c_0c_1 \cdots$ is accepting.
    \item $\pi\in\Win\B$: If there exists an~accepting sequence of transitions
	$\vec\delta^\B = \delta_0^\B \delta_1^\B \cdots \in \Delta^\B(b)$
	s.t.~for every $i\in\omega$ we have $f_i(\delta^\B_i) = d_i$, then $c_0c_1 \cdots$ is rejecting.
\end{winnize}

The following lemma states that the separability game correctly characterises the deterministic separability problem.
\begin{lemma}
\label{lem:CDeterministicSeparability:correctness}
\Separator wins $\CDeterministicSeparabilityGame \A \B C$
if, and only if,
$\lang \A$, $\lang \B$ can be separated by a~deterministic parity tree automaton with priorities in $C$.
\end{lemma}

\noindent
We present a full proof in order to show the r\^ole of \Separator's selectors.

\begin{proof}[Soundness]
Assume that \Separator wins the separability game $G := \CDeterministicSeparabilityGame \A \B C$ by a~finite-memory winning strategy $\M = \tuple{ M, \ell_0, (\overline{c}, \overline{f}), \tau}$.
Strategy $\M$ has two decision functions:
$\overline{c}$ assigns to each $\ell\in M$ a~priority $\overline{c}(\ell)\in C$,
and $\overline{f}$ assigns to each $\ell\in M$ and $a\in\Sigma$ a~selector $\overline{f}(\ell,a)\in\set{\L, \R}^{\Delta^\B(a)}$. Moreover, the type of the memory update function is $\tau\colon M\times\Sigma\times \set{\L, \R}\to M$.
Consider a~deterministic parity tree automaton
$\S := \tuple{ \Sigma, M, \ell_0, \Omega^\S, \Delta^\S}$
which has the same set of states $M$ and initial state $\ell_0$ as $\M$,
priorities are induced by the decision function $\overline{c}$ of $\M$ as
$\Omega^\S(\ell) := \overline{c}(\ell)$,
and transitions are of the form
    $\Delta^\S = \setof{\tuple{\ell, a, \tau(\ell, a, \L), \tau(\ell, a, \R)}}{\ell \in M, a \in \Sigma}$.

We show that $\S$ separates $\lang \A$, $\lang \B$.
We first show $\lang \A \subseteq \lang \S$.
Let $t\in\lang\A$ be a~tree that is accepted by the automaton $\A$, as witnessed by an~accepting run $\rho^\A$. Let $\rho^\S$ be the unique run of $\S$ over $t$. Consider any branch $d_0d_1\cdots\in\set{\L,\R}^\omega$.
We need to show that the sequence of priorities $\big(\Omega^\S(\rho^\S(d_0\cdots d_{i-1}))\big)_{i\in\omega}$ is accepting.
Consider a~play $\pi$ of $\Game{}$ where at the $i$-th round
\Separator plays according to the strategy $\M$ with current memory state $\ell_i\in M$
and \Input plays according to the letters from $t$ and directions $d_0d_1\cdots$ fixed above:
\begin{simuize}
\item[\Choice{S}{c}] \Separator plays the priority $c_i:=\overline{c}(\ell_i)\in C$.
\item[\Choice{I}{a}] \Input plays the letter $a_i:=t(u_i)\in\Sigma$, where $u_i:=d_0\cdots d_{i-1}$.
\item[\Choice{S}{f}] \Separator plays the selector $f_i:=\overline{f}(\ell_i,a_i) \in \set{\L, \R}^{\Delta^\B(a_i)}$
(the selector is irrelevant in this part of the proof).
\item[\Choice{I}{d}] \Input plays the direction $d_i \in \set{\L, \R}$ as fixed above.
\end{simuize}
The next memory state is $\ell_{i+1}:= \tau(\ell_i, a_i, d_i)$. 
Let the resulting infinite play be $\pi = \tuple{c_0, a_0, f_0, d_0} \tuple{c_1, a_1, f_1, d_1} \cdots$.
By the construction of $\S$ we know that $\ell_i= \rho^\S(u_i)$ and therefore $c_i=\Omega^\S(\rho^\S(u_i))$.
Since $t\in\lang\A$, there exists an accepting sequence of transitions $\vec\delta^\A = \delta_0^\A \delta_1^\A \cdots \in \Delta^\A(b)$
along the path $b = \tuple{a_0, d_0} \tuple{a_1, d_1} \cdots$.
Since \Separator is winning,
$\pi \in \Win\A$ and thus the sequence $c_0c_1\cdots$ is accepting, as required.

We now argue that $\lang \S$ and $\lang \B$ are disjoint.
Towards reaching a~contradiction, assume that $t\in\lang\S\cap\lang\B$ belongs to their intersection.
Let $\rho^\S$ be the unique run of $\S$ over $t$,
and let $\rho^\B$ be an~accepting run of $\B$ over $t$.
Consider a~play $\pi = \tuple{c_0, a_0, f_0, d_0} \tuple{c_1, a_1, f_1, d_1} \cdots$ of $\Game{}$
where the $i$-th round is played as above
except that \Input plays the direction $d_i := f_i(\delta^\B_i)$,
obtained by applying the selector $f_i$
to the transition $\delta^\B_i := \big(\rho^\B(u_i), t(u_i), \rho^\B(u_i\L),\allowbreak \rho^\B(u_i\R)\big)$
determined according to the run $\rho^\B$.
By the choice of directions $d_i$'s,
the sequence of transitions $\vec\delta^\B = \delta_0^\B\delta_1^\B\cdots \in (\Delta^\B)^\omega$ satisfies $f_i(\delta^\B_i)=d_i$ for every $i\in\omega$.
Since the run $\rho^\B$ is accepting,
$\vec\delta^\B$ is accepting.
Since \Separator is winning, $\pi \in \Win\B$ and thus the sequence of priorities $c_0c_1 \cdots$ is rejecting.
However, this is a~contradiction, because for each $i\in\omega$ we have $\ell_i=\rho^\S(u_i)$ and $c_i=\Omega^\S(\ell_i)$ and we assumed that the run $\rho^\S$ is accepting.
\end{proof}

\begin{proof}[Completeness]
Assume that $\S=\tuple{ \Sigma, Q^\S, q_0^\S, \Delta^\S, \Omega^\S}$ is a~deterministic automaton with priorities in $C$ separating $\lang \A$, $\lang \B$,
and we show that \Separator wins the separability game $G$.
Since $\S$ is a separator, we have that $\lang \S \bot \lang \B$, 
and by \cref{cor:disjoint-pathfinder} there exists a~pathfinder $\PP$.
Consider the following strategy of \Separator, with memory structure $Q^\S$ and initial memory state $q_0^\S$.
At the $i$-th round of $\Game{}$, starting with a~memory state $q^\S_i$,
\begin{simuize}
\item[\Choice{S}{c}] \Separator plays the priority $c_i:=\Omega^\S(q^\S_i)\in C$.
\item[\Choice{I}{a}] \Input plays an arbitrary letter $a_i\in\Sigma$.
\item[\Choice{S}{f}] \Separator plays the selector $f_i:=\PP(\delta^\S_i, \_)\in \set{\L, \R}^{\Delta^\B(a_i)}$, where $\Delta^\S(q^\S_i,a_i)=\set{\delta^\S_i}$.
\item[\Choice{I}{d}] \Input plays an arbitrary direction $d_i\in \set{\L, \R}$.
\end{simuize}
The next memory state is $q^\S_{i+1}:= q^\S_{d_i,i}$, where $\delta^\S_i=\tuple{q^\S_i,a_i,q^\S_{\L,i},q^\S_{\R,i}}$.
This concludes the description of the $i$-th round of $G$.
Let the resulting infinite play be $\pi = \tuple{c_0, a_0, f_0, d_0}\allowbreak \tuple{c_1, a_1, f_1, d_1} \cdots$,
with induced infinite path $b := \tuple{a_0, d_0} \tuple{a_1, d_1} \cdots$. Let $\vec\delta^\S:=\delta^\S_0\delta^\S_1\cdots$ be the sequence of transitions used to define the selectors $f_i$. Clearly $\vec\delta^\S\in\Delta^\S(b)$.

First, we argue that $\pi \in \Win\A$ holds.
Let $\vec\delta^\A = \delta_0^\A \delta_1^\A \cdots \in \Delta^\A(b)$ be an~accepting sequence of transitions of the automaton $\A$.
Since each state of $\A$ is productive, one can construct a~tree $t\in\lang\A$ s.t.~$b\in\paths{t}$.
Since $\lang \A \subseteq \lang \S$ by the assumption, $t \in \lang \S$ as well,
and since $\S$ is deterministic, the unique run of $\S$ over $t$ is accepting.
By the definition of \Separator's strategy, the sequence of priorities along the branch $d_0d_1\cdots$ of this accepting run is precisely $c_0c_1 \cdots$, which thus must be accepting, as required.

Regarding $\Win\B$, let $\vec\delta^\B := \delta^\B_0 \delta^\B_1 \cdots \in \Delta^\B(b)$ be an~accepting sequence of transitions over the path $b$
conform to the selectors $f_i$, i.e., for every $i\in\omega$ we have $f_i(\delta_i^\B)=d_i$.
By the definition of $f_i$,
for every $i\in\omega$ we have $d_i=\PP(\delta^\S_i,\delta^\B_i)$.
Thus, the assumptions of \cref{rem:pathfinder-property} are satisfied and at least one of the sequences $\vec\delta^\S$, $\vec\delta^\B$ must be rejecting. Since we assumed that $\vec\delta^\B$ is accepting, it means that $\vec\delta^\S$ is rejecting, and so is $c_0c_1 \cdots$ since $c_i=\Omega^\S(\delta^S_i)$.
\end{proof}

\subsection{A variant of the separability game for trees}
\label{app:naive-det-gen}

A first attempt at generalising the case of $\omega$-words to infinite trees
is to let \Input play a \emph{direction} $d_i \in \set{\L, \R}$ after she plays a letter $a_i$,
and not considering selectors for \Separator.
This yields the following simpler variant of the game considered at the beginning of this section. 
At round $i$ of the separability game,

\begin{gameize}
    \item[\Choice{S}{c}] \Separator plays a priority $c_i \in C$.
    \item[\Choice{I}{a}] \Input plays a letter $a_i \in \Sigma$.
	\item[\Choice{I}{d}] \Input plays a direction $d_i \in \set{\L, \R}$.
\end{gameize}
%
\Separator wins the corresponding infinite play
$\pi = \tuple{c_0, a_0, d_0} \tuple{c_1, a_1, d_1} \cdots \in (C \times \Sigma \times \set{\L, \R})^\omega$
if the induced infinite path $b = \tuple{a_0, d_0} \tuple{a_1, d_1} \cdots \in (\Sigma \times \set{\L, \R})^\omega$
satisfies the following two conditions:
\begin{itemize}
    \item $\pi\in\Win\A$: If there exists an~accepting sequence of transitions
    $\vec\delta^\A = \delta_0^\A \delta_1^\A \cdots \in \Delta^\A(b)$, then $c_0c_1 \cdots$ is accepting.
    \item $\pi\in\Win\B$: If there exists an~accepting sequence of transitions
    $\vec\delta^\B = \delta_0^\B \delta_1^\B \cdots \in \Delta^\B(b)$, then $c_0c_1 \cdots$ is rejecting.
\end{itemize}
It turns out that the winning condition above is not strong enough in order to characterise deterministic separability over languages of infinite trees.
A~deterministic automaton $\S$ is \emph{universally rejecting} on a set of trees $L$
if, for every $t \in L$, \emph{all} branches in the corresponding run in $\S$ are rejecting.
The following lemma states that the game in this section characterises separability by deterministic automata $\S$ which are universally rejecting on $\lang \B$.

\begin{lemma}
    \Separator wins the game above if, and only if,
    $\lang \A$, $\lang \B$ can be separated by a deterministic separator $\S$ with priorities from $C$
    which is universally rejecting on $\lang \B$.
\end{lemma}

\begin{proof}[Proof sketch]
    The proof is analogous to that of \cref{lem:CDeterministicSeparability:correctness}.
\end{proof}

\section{Separability by deterministic automata}
\label{sec:det_gen}


%

In this section we present a game-theoretic characterisation of the deterministic separability problem.
Notice that here we do not fix in advance a finite set of priorities $C$.
The deterministic-separability game $\DeterministicSeparabilityGame \A \B$ below
is a variant of the game with fixed priorities $C$ from \cref{sec:det_ind}.

\begin{gamebox}{Deterministic-separability game $\DeterministicSeparabilityGame \A \B$}
    At the $i$-th round:
\begin{gameize}
\item[\Choice{I}{a}] \Input plays a~letter $a_i \in \Sigma$.
\item[\Choice{S}{f}] \Separator plays a~selector $f_i \in \set{\L, \R}^{\Delta^\B(a_i)}$.
\item[\Choice{I}{d}] \Input plays a~direction $d_i \in \set{\L, \R}$.
\end{gameize}
\end{gamebox}

\noindent
\Separator wins the resulting infinite play $\pi = \tuple{a_0, f_0, d_0} \tuple{a_1, f_1, d_1} \cdots$,
with induced infinite path $b := \tuple{a_0, d_0} \tuple{a_1, d_1} \cdots$,
if at least one of the two conditions below fails:
\begin{winnize}
\item $\pi\in\Win\A$: There exists an~accepting sequence of transitions
    $\vec\delta^\A = \delta_0^\A \delta_1^\A \cdots \in \Delta^\A(b)$.
\item $\pi\in\Win\B$: There exists an~accepting sequence of transitions
    $\vec\delta^\B = \delta_0^\B \delta_1^\B \cdots \in \Delta^\B(b)$
	s.t.~for every $i \in \omega$ we have $f_i(\delta_i^\B) = d_i$.
\end{winnize}

Before we prove the equivalence between the game and the existence of a separator,
we define a separator candidate, namely the path-closure of $\lang\A$.
This is important since it will turn out that if a separator exists,
then the path-closure is itself a separator.
Given a~language of trees $L$,
its \emph{path-closure}, denoted $\pathclosure L$,
is the set of all trees $t$ s.t.~for every path $b \in \paths t$
there exists some tree $t' \in L$ s.t.~$b \in \paths {t'}$ as well.

The following lemma states formally some basic facts
justifying that $\pathclosure \_$ is indeed a~closure operator.

\begin{lemma}
    \label{lem:path-closure}\mbox{}
    \begin{enumerate}
        \item The path-closure operator is monotonic (w.r.t.~set inclusion): \\
        If $L \subseteq M$, then $\pathclosure L \subseteq \pathclosure M$.
        
        \item The path-closure operator is non-decreasing (w.r.t.~set inclusion): \\
        For any language $L$, $L \subseteq \pathclosure L$.

        \item The path-closure $\pathclosure L$ of $L$ is the smallest (w.r.t.~set inclusion) language of infinite trees $M$
        s.t.~a) $M$ contains $L$: $L \subseteq M$,
        and b) $M$ is path-closed: $\pathclosure M \subseteq M$.
    \end{enumerate}
\end{lemma}

\begin{proof}
    The first two properties are clear.
    For the third property, $M := \pathclosure L$ itself satisfies
    a) since the path-closure operator is non-increasing,
    and b) since the path-closure operator is idempotent $\pathclosure {\pathclosure L} = \pathclosure L$.
    Now let $M$ be an arbitrary language s.t.~a) $L \subseteq M$,
    and b) $\pathclosure M \subseteq M$.
    We immediately have
    \begin{align*}
        L \stackrel {\text{\tiny (1)}} \subseteq \pathclosure L
          \stackrel {\text{\tiny (2)}} \subseteq \pathclosure M
          \stackrel {\text{\tiny (3)}} \subseteq M,
    \end{align*}
    where (1) follows from the fact that the path-closure operator is non-decreasing,
    (2) from the fact that it is monotone,
    and (3) from the fact that $M$ is path-closed.

Now consider any deterministic tree automaton $\B$. Observe that since $\B$ is deterministic, we clearly get $\pathclosure{\lang\B}=\lang\B$.
Now assume that $\lang\A\subseteq\lang\B$. By monotonicity of $\pathclosure\_$ we obtain that $\pathclosure{\lang\A}\subseteq \pathclosure{\lang\B}=\lang\B$.
\end{proof}

%
%
The path-closure operator is directly connected with deterministic automata.

\begin{restatable}[\protect{c.f.~\cite[Proposition~1]{NiwinskiWalukiewicz:TCS:2003}}]{lemma}{lemPathClosureAuto}
\label{lem:path-closure-auto}
Given a~nondeterministic automaton $\A$ one can construct a~deterministic automaton
$\pathautomaton \A$ recognising the path closure of $\lang\A$, i.e.,~$\lang {\pathautomaton A} = \pathclosure{\lang \A}$.
Moreover,
$\lang {\pathautomaton A}$ is the smallest deterministic language containing $\lang\A$.
\end{restatable}

\begin{proof}
Fix a~nondeterministic automaton $\A = \tuple{\Sigma, Q, q_0, \Omega, \Delta}$. Our aim is to construct a~deterministic automaton
$\pathautomaton \A$ recognising the path closure if $\lang\A$, i.e.,~$\lang {\pathautomaton A} = \pathclosure{\lang \A}$.
Let $\D=\tuple{ \Sigma\times \set{\L,\R}, Q^\D, q_0^\D, \Delta^\D, \Omega^\D} $ be a~deterministic parity automaton over $\omega$-words over the alphabet $\Sigma\times \set{\L,\R}$ that recognises the set of paths $b$ s.t.~there exists an~accepting sequence of transitions in $\Delta^\A(b)$.
It is easy to see how a nondeterministic such automaton can be obtained,
and it can be determinised thanks to \cref{lem:NPA2DPA}.
Consider a~deterministic parity tree automaton
$\pathautomaton \A := \tuple{ \Sigma, Q^\D, q_0^\D, \Omega^\D, \Delta'}$
which has the same set of states, initial state, and priority mapping as $\D$,
and transitions are of the form
\begin{align*}
    \Delta' = \setof{\tuple{q, a, \Delta^\D(q, (a, \L)), \Delta^\D(q, (a, \R))}}{q \in Q^\D, a \in \Sigma}.
\end{align*}

Now, we claim that the following conditions are equivalent, for a~tree $t\in\trees_\Sigma$:
\begin{enumerate}
\item $t\in \pathclosure{\lang \A}$,
\item for every path $b\in\paths{t}$ there exists a~tree $t'\in\lang\A$ s.t.~$b\in\paths{t'}$,
\item for every path $b\in\paths{t}$ there exists an~accepting sequence of transitions in $\Delta^\A(b)$,
\item for every path $b\in\paths{t}$ the automaton $\D$ accepts $b$,
\item $t\in \lang{\pathautomaton \A}$.
\end{enumerate}
Indeed, the only nontrivial implication is ``$3 \Rightarrow 2$'', however, since every state of $\A$ is productive, one can easily construct the tree $t'$ by extending the considered sequence of transitions of $\A$ to the subtrees outside the path $b$.
We can thus conclude $\pathclosure{\lang \A} = \lang {\pathautomaton \A}$, as required.
\end{proof}

The following lemma binds together the game $\DeterministicSeparabilityGame \A \B$, separability, and path-closures.

\begin{restatable}{lemma}{lemDetSep}
    \label{lem:det:sep}
    The following three conditions are equivalent:
    \begin{enumerate}
        \item \Separator wins the deterministic-separability game $\DeterministicSeparabilityGame \A \B$.
        \item The automaton $\pathautomaton \A$ is a deterministic separator for $\lang \A$, $\lang \B$.
        \item There exists a~deterministic separator for $\lang \A$, $\lang \B$.
    \end{enumerate}
\end{restatable}


\begin{proof}

    We begin by proving ``$1 \Rightarrow 2$''.
    Assume that \Separator wins the separability game $\DeterministicSeparabilityGame \A \B$ by a~finite-memory winning strategy $\M = \tuple{ M, \ell_0, \overline{f}, \tau}$ which has one decision function $\overline{f}$ assigning to each $\ell\in M$ and $a\in\Sigma$ a~selector $\overline{f}(\ell,a)\in\set{\L, \R}^{\Delta^\B(a)}$. Moreover, the type of the memory update function is $\tau\colon M\times\Sigma\times \set{\L, \R}\to M$.

    Let $\S:=\pathautomaton \A$ be the deterministic automaton recognising the path closure of $\lang\A$. We show that $\S$ separates $\lang \A$, $\lang \B$.

    The condition $\lang \A \subseteq \lang \S$ follows immediately from the fact that $\lang{\S}=\pathclosure{\lang\A}$ and by Item~1 of \cref{lem:path-closure} we know that $\pathclosure{\lang\A}\supseteq\lang\A$.

    We now argue that $\lang \S$ and $\lang \B$ are disjoint.
    Towards reaching a~contradiction, assume that $t\in\lang\S\cap\lang\B$ belongs to their intersection. Let $\rho^\S$ be the unique run of $\S$ over $t$, and let $\rho^\B$ be an~accepting run of $\B$ over $t$. Consider a~play $\pi$ of $\Game{}:=\DeterministicSeparabilityGame \A \B$ where at the $i$-th round, \Separator plays according to the strategy $\M$ with current memory state $\ell_i\in M$ and \Input plays as follows:
    \begin{simuize}
    \item[\Choice{I}{a}] \Input plays the letter $a_i:=t(u_i)\in\Sigma$, where $u_i:=d_0\cdots d_{i-1}$.
    \item[\Choice{S}{f}] \Separator plays the selector $f_i:=\overline{f}(\ell_i,a_i) \in \set{\L, \R}^{\Delta^\B(a_i)}$.
    \item[\Choice{I}{d}] \Input plays the direction $d_i:=f_i(\delta^\B_i) \in \set{\L, \R}$, where $\delta^\B_i$ is the respective transition of $\rho^\B$, i.e.,~$\delta^\B_i:=\big(\rho^\B(u_i), t(u_i), \rho^\B(u_i\L), \rho^\B(u_i\R)\big)$.
    \end{simuize}
    The next memory state is $\ell_{i+1}:= \tau(\ell_i, a_i, d_i)$. Let $\pi$ be the obtained play. By the choice of directions $d_i$ we know that the sequence of transitions $\vec\delta^\B = \delta_0^\B\delta_1^\B\cdots \in (\Delta^\B)^\omega$ satisfies $f_i(\delta^\B_i)=d_i$ for every $i\in\omega$. Moreover, as the run $\rho^\B$ is accepting, we know that $\vec\delta^\B$ is accepting. Therefore, $\Win\B$ holds for this play. It means that $\Win\A$ must fail, meaning that the infinite path $b:=\tuple{a_0,d_0}\tuple{a_1,d_1}\cdots$ does not belong to $\paths{\lang\A}$. However, this is a contradiction with the assumption that the run $\rho^\S$ of $\S$ over $t$ is accepting.

    The implication ``$2 \Rightarrow 3$'' is trivial.

    Finally, we prove ``$3 \Rightarrow 1$''.
    Assume that $\S=\tuple{ \Sigma, Q, q_0, \Delta, \Omega}$ is a~deterministic automaton separating $\lang \A$, $\lang \B$,
    and we show that \Separator wins the separability game $\DeterministicSeparabilityGame \A \B$.
    Since $\S$ is a~separator, we have that $\lang \S \bot \lang \B$, 
    and thus \Pathfinder wins the disjointness game $\DisjointnessGame \S \B$, see \cref{cor:disjoint-pathfinder}. Let $\PP$ be a~pathfinder as in \cref{ssec:disjointness-game}.

    Consider the following strategy of \Separator, with the memory structure $Q^\S$ and the initial memory state $q_0^\S$. At the $i$-th round of $\Game{}:=\DeterministicSeparabilityGame \A \B$, starting with a~memory state $q^\S_i$, \Separator plays as follows:
    \begin{gameize}
    \item[\Choice{I}{a}] \Input plays an arbitrary letter $a_i\in\Sigma$.
    \item[\Choice{S}{f}] \Separator plays the selector $f_i:=\PP(\delta^\S_i, \_)\in \set{\L, \R}^{\Delta^\B(a_i)}$, where $\Delta^\S(q^\S_i,a_i)=\set{\delta^\S_i}$.
    \item[\Choice{I}{d}] \Input plays an arbitrary direction $d_i\in \set{\L, \R}$.
    \end{gameize}
    The next memory state is $q^\S_{i+1}:= q^\S_{d,i}$, where $\delta^\S_i=\tuple{q^\S_i,a_i,q^\S_{\L,i},q^\S_{\R,i}}$.

    Let the resulting infinite play be $\pi$, with the induced infinite path $b := \tuple{a_0, d_0} \tuple{a_1, d_1} \cdots$. Let $\vec\delta^\S:=\delta^\S_0\delta^\S_1\cdots$ be the sequence of transitions used to define the selectors $f_i$.

    Assume for the sake of contradiction that both $\Win\A$ and~$\Win\B$ hold, as witnessed by accepting sequences of transitions $\vec\delta^\A = \delta_0^\A \delta_1^\A \cdots \in \Delta^\A(b)$ and $\vec\delta^\B = \delta_0^\B \delta_1^\B \cdots \in \Delta^\B(b)$. The sequence $\vec\delta^\A$ implies that there exists a~tree $t\in\trees_\Sigma$ s.t.~$b$ is a path of $t$ and $t\in\lang\A$. Since $\lang\A\subseteq\lang\S$, the run of the automaton $\S$ must be accepting on $t$ and therefore the sequence of transitions $\vec\delta^\S$ is accepting. Moreover, the choice of the selectors $f_i$ means that for every $i\in\omega$ we have $d_i=\PP(\delta^\S_i, \delta^\B_i)$. Thus, the assumptions of \cref{rem:pathfinder-property} are satisfied and at least one of the sequences $\vec\delta^\S_i$, $\vec\delta^\B_i$ must be rejecting. A~contradiction, because we assumed that $\vec\delta^\B$ is accepting and we know that $\vec\delta^\S$ is also accepting.
\end{proof}


\section{Separability by game automata}
\label{sec:game_gen}

In this section we provide a game-theoretic characterisation for the game automata separability problem.
Fix two automata $\A$ and $\B$ and consider the following separability game $\GameSeparabilityGame \A \B$.
The new ingredient is that \Separator can choose a \emph{mode}---a symbol from the set $\set{\lor,\land}$.
It has two uses. First, in the construction of the separating game automaton, the mode dictates whether there will be a conjunctive or a disjunctive transition.
Second, depending on the chosen mode, \Separator will have to play a~selector for the automaton $\A$ or $\B$,
which will guarantee that the constructed automaton is a separator.

\begin{gamebox}{Game-separability game $\GameSeparabilityGame \A \B$}
    At the $i$-th round:
    \begin{gameize}
        \item[\Choice{I}{a}] \Input plays a~letter $a_i \in \Sigma$.
        \item[\Choice{S}{m}] \Separator plays a~mode $m_i\in\set{\lor,\land}$.
        \item[\Choice{S}{f}] \Separator plays either
        \begin{enumerate}
            \item a~selector $f_i \in \set{\L, \R}^{\Delta^\A(a_i)}$ for $\A$ if $m_i=\lor$ or
            \item a~selector $f_i \in \set{\L, \R}^{\Delta^\B(a_i)}$ for $\B$ if $m_i=\land$.
        \end{enumerate}
        \item[\Choice{I}{d}] \Input plays a~direction $d_i \in \set{\L, \R}$.
    \end{gameize}
\end{gamebox}

\noindent
\Separator wins an~infinite play $\pi = \tuple{a_0,m_0,f_0,d_0}\tuple{a_1,m_1,f_1,d_1}\cdots$
inducing a~path $b = \tuple{a_0,d_0}\tuple{a_1,d_1} \cdots$
whenever at least one of the two conditions below fail:
\begin{winnize}
    \item $\pi\in \Win\A$:
    There exists an accepting sequence of transitions
    $\vec\delta^\A = \delta_0^\A \delta_1^\A \cdots \in \Delta^\A(b)$
    s.t.~for all $i \in \N$ we have $(m_i = \lor)\Rightarrow f_i(\delta_i^\A) = d_i$.

    \item $\pi\in \Win\B$:
    There exists an accepting sequence of transitions
    $\vec\delta^\B = \delta_0^\B \delta_1^\B \cdots \in \Delta^\B(b)$
    s.t.~for all $i \in \N$ we have $(m_i = \land) \Rightarrow f_i(\delta_i^\B) = d_i$.
\end{winnize}

\begin{lemma}
    \label{lem:game:sep}
    \Separator wins the separability game $\GameSeparabilityGame \A \B$
    if, and only if, there exists a game automaton $\S$ separating $\lang \A$, $\lang \B$.
\end{lemma}

In the proof of this lemma we will build separating automata
with a more general acceptance condition than the parity condition,
which will simplify the technical details.
%
%
A~\emph{generalised game automaton} $\A=\tuple{\Sigma, Q, q_0, \Delta, \D}$ is just
like a~game automaton except that the priority mapping $\Omega$
is replaced by a~deterministic $\omega$-word parity automaton $\D$ over alphabet $\Sigma\times \set{\L,\R}$.
A~run $\rho\in\trees_Q$ of such an~automaton over a~tree $t\in\trees_\Sigma$ is \emph{accepting}
if for every path $b=\tuple{a_0, d_0} \tuple{a_1, d_1} \cdots \in \paths{t}$
either $\rho(d_0\cdots d_{i-1})=\top$ for some $i\in\omega$, or $b\in\lang\D$.
The acceptance game $\AcceptanceGame \A t$ can easily be adapted to the case of a~generalised game automaton $\A$
by only modifying the winning condition.

\begin{restatable}{lemma}{lemGenGame}
    \label{lem:generalised:tree:automata}
    A~generalised game automaton $\A$
    with a~generalised acceptance condition recognised by a~deterministic parity automaton $\D$
    can be transformed into an~equivalent (ordinary) game automaton $\B$
    of size polynomial in $\A$ and $\D$.
    %
\end{restatable}

\begin{proof}
    Consider the game automaton $\B$ defined as the following product of $\A$ and $\D$:
    \[\B:=\tuple{\Sigma, (Q^\A\setminus\set{\top})\times Q^\D\cup\{\top\}, (q_0^\A,q_0^\D), \Delta^\B, \Omega^\B},\]
    where $\Omega^\B$ is just inherited from $\D$, i.e.,~$\Omega^\B(q^\A,q^\D):=\Omega^\D(q^\D)$.
    Moreover, for each conjunctive $\A$-transition $\tuple{q^\A,a,q^\A_\L,q^\A_\R}\in\Delta^\A$,
    $\Delta^\B$ contains the transition
    \[\big((q^\A,q^\D), a, (q^\A_\L,q^\D_\L), (q^\A_\R,q^\D_\R)\big),\]
    where for $d \in \set{\L,\R}$ we have $\Delta^\D\big(q^\D, \tuple{a,d}\big)=q^\D_d$.
    Similarly, for each disjunctive $\A$-transition $\tuple{q^\A,a,q^\A_\L,\top}\in\Delta^\A$ (resp., $\tuple{q^\A,a,\top,q^\A_\R}\in\Delta^\A$), $\Delta^\B$ contains the transition $\big((q^\A,q^\D), a, (q^\A_\L,q^\D_\L), \top\big)$ (resp., $\big((q^\A,q^\D), a, \top, (q^\A_\R,q^\D_\R)\big)$), where for $d \in \set{\L,\R}$ we have $\Delta^\D\big(q^\D, \tuple{a,d}\big)=q^\D_d$.

    Notice that for every tree $t\in\trees_\Sigma$ there is a~bijection between the runs $\rho^\A$ of $\A$ over $t$ and runs $\rho^\B$ of $\B$ over $t$: given a~run $\rho^\B$ we can just project it onto the first coordinate to obtain $\rho^\A$, and the run $\rho^\B$ is obtained in a~top-down deterministic way from $\rho^\A$ by running the automaton $\D$ deterministically on all the paths. Therefore, it is enough to argue that if $\rho^\A$ and $\rho^\B$ are two such runs then $\rho^\A$ is accepting if and only if $\rho^\B$ is. Consider a~branch $d_0d_1\cdots\in\set{\L,\R}$ and let $u_i:=d_0\cdots d_{i-1}$ for $i\in\omega$. Without loss of generality assume that $\rho^\A(u_i)\neq\top$ for every $i\in\omega$ (otherwise both runs are accepting on this branch). For each $i\in\omega$ let $(q^\A_i,q^\D_i):= \rho^B(u_i)$ and notice that by the choice of the runs $\rho^\A$ and $\rho^\B$ we know that $\rho^\A(u_i)=q^\A_i$. Now let $b:=\tuple{t(u_0), d_0}\tuple{t(u_1), d_1}\cdots$ be the path used to define the generalised acceptance condition of $\A$ on the considered branch. By the construction of the automaton $\B$, we know that the sequence of states $q^\D_0q^\D_1\cdots$ is the run of $\D$ on $b$. Therefore, $\rho^\A$ satisfies the generalised acceptance condition on the path $b$ if and only if $\rho^\B$ satisfies the parity condition on the branch $d_0d_1\cdots$.
\end{proof}

We now prove \cref{lem:game:sep}.

\begin{proof}[Soundness]

    Assume that \Separator wins the game-separability game above $\Game{} := \GameSeparabilityGame \A \B$ and we show that there exists a~game automaton $\S$ separating $\lang \A$ from $\lang \B$.
    Let $\M = \tuple{ M, \ell_0, (\overline{m},\overline{f}), \tau}$ be a~finite-memory winning strategy of \Separator in $\Game{}$.

    Before we move to the construction of the separating automaton,
    we first define its generalised acceptance condition.
    Let $L_\A$ (resp., $L_\B$) be the set of those paths
    $b=\tuple{a_0,d_0}\tuple{a_1,d_1}\cdots\in (\Sigma\times\set{\L,\R})^\omega$
    s.t.~the unique play $\pi$ of $\Game {}$ in which \Input plays consecutive letters and directions from $b$ and \Separator uses her winning strategy $\M$, satisfies the condition $\Win\A$ (resp., $\Win\B$).
    Since the strategy $\M$ is winning for \Separator, the languages $L_\A$ and $L_\B$ are disjoint.
    Moreover, since the strategy $\M$ is finite memory and both $\Win\A$, $\Win\B$ are $\omega$-regular, so are the languages $L_\A$ and $L_\B$. Let $\D$ be any deterministic automaton over $\omega$-words that separates $L_\A$ from $L_\B$ (the simplest case is to take $\D$ recognising the language $L_\A$).
    We build a~separating automaton as a~generalised game automaton
    \begin{align*}
        \S &:= \alpha(\M,\D):=\tuple{\Sigma, M\cup\{\top\}, \ell_0, \Delta^\S, \D}, \textrm{ where} \\
        \Delta^\S\big(\ell, a\big) &:=
            \begin{cases}
                \set{\tuple{\ell,a,\ell_\L, \top}, \tuple{\ell,a,\top,\ell_\R}}
                &\text{if $\overline{m}(\ell, a) = \lor $,}\\
                \set{\tuple{\ell,a,\ell_\L, \ell_\R}}
                & \text{if $\overline{m}(\ell, a) = \land$,}
            \end{cases}
    \end{align*}
    for every $\ell \in M$ and $a \in \Sigma$,
    where for $d\in\set{\L,\R}$ we have $\ell_d := \tau(\ell, a, d)$.
    %
    We now show that $\S$ separates $\lang \A$ from $\lang \B$.
    In order to show $\lang \A \subseteq \lang \S$, let $t \in \lang \A$
    as witnessed by an accepting run $\rho^\A$.
    We show that \Automaton wins the acceptance game $\Game\S := \AcceptanceGame \S t$.
    To show this we play in parallel the separability game $\Game{}$ and the acceptance game $\Game\S$, maintaining the following invariant:
    At the $i$-th round, the current finite path of the input tree $t$ is $\tuple{a_0,d_0} \cdots \tuple{a_{i-1}, d_{i-1}}$,
    \Separator's winning strategy $\M$ in the separability game $\Game{}$ is in memory state $\ell_i$,
    the current state of the separating automaton $\S$ in the acceptance game $\Game\S$ is also $\ell_i$,
    and $\rho^\A(d_0\cdots d_{i-1})=q_i^\A$.
    The $i$-th round is then played as follows:
    \begin{simuize}
    \item[\GChoice{\Game{}}{I}{a}] \Input plays the letter $a_i := t(u_i)$ for $u_i:=d_0 \cdots d_{i-1}$.
    \item[\GChoice{\Game{}}{S}{m}] \Separator plays the~mode $m_i:= \overline{m}(\ell_i,a_i) \in\set{\lor,\land}$.
    \item[\GChoice{\Game{}}{S}{f}] \Separator plays either
            \begin{enumerate}
                \item a~selector $f_i :=\overline{f}(\ell_i,a_i) \in \set{\L, \R}^{\Delta^\A(a_i)}$ for $\A$ if $m_i=\lor$ or
                \item a~selector $f_i :=\overline{f}(\ell_i,a_i) \in \set{\L, \R}^{\Delta^\B(a_i)}$ for $\B$ if $m_i=\land$.
            \end{enumerate}
        \item[\GChoice{\Game\S}{A}{\delta}] \Automaton plays the transition $\delta^\S_i\in\Delta^\S(\ell_i, a_i)$, defined as follows.
        Let $\delta^\A_i:= \big(\rho^\A(u_i), t(u_i),\rho^\A(u_i\L),\rho^\A(u_i\R)\big)$ be the $\A$-transition used in $u_i$ by the run $\rho^\A$. We distinguish two cases.

        \begin{enumerate}
        \item In the first case, assume that \Separator played $m_i=\lor$ and $f_i \in \set{\L, \R}^{\Delta^\A(a_i)}$.
        It means that $\Delta^\S\big((\ell_i,q_i), a_i\big)$ contains two disjunctive transitions,
        $\delta^\S_{\L,i}:=\tuple{\ell_i,a_i,\ell_{\L,i},\top}$ and
        $\delta^\S_{\R,i}:=\tuple{\ell_i,a_i,\top,\ell_{\R,i}}$. Let us put $\delta^\S_i:=\delta^\S_{f_i(\delta^\A_i), i}$, i.e.,~the transition that sends a~non-$\top$ state in the direction given by $f_i(\delta^\A_i)$.
        
        \item In the second case, \Separator played $m_i=\land$ and $f_i \in \set{\L, \R}^{\Delta^\B(a_i)}$.
        It means that $\Delta^\S\big(\ell_i, a_i\big)$ contains one conjunctive transition $\delta^\S_i:=\tuple{\ell_i,a_i,\ell_{\L,i},\ell_{\R,i}}$.
        \end{enumerate}

        \item[\GChoice{\Game\S}{I}{d}] \Input plays an~arbitrary direction $d_i\in\set{\L,\R}$.
        \item[\GChoice{\Game{}}{I}{d}] \Input plays the direction $d_i\in\set{\L,\R}$.
    \end{simuize}

    If $m_i = \lor$ and $d_i\neq f_i(\delta_i^\A)$
    then the next position of the acceptance game $\Game\S$ is $(u_id_i, \top)$, which is a~winning position for \Automaton.
    Therefore, w.l.o.g. we assume that:
    \begin{equation}
    \label{eq:lor-to-direction}
        \forall i\in\omega.\ (m_i = \lor) \Rightarrow f_i(\delta_i^\A)=d_i.
    \end{equation}
    Moreover, the new state of $\S$ in $\Game\S$ is $\ell_{i+1} := \tau(\ell_i, a_i, d_i)$.
    Similarly, the new memory state of $\M$ in $\Game{}$ is $\ell_{i+1}$.
    This concludes the description of the $i$\=/th round of both games. Clearly the invariant is preserved.
    We argue that \Automaton wins the resulting infinite play
    $\tuple{\delta_0^\S, d_0}\tuple{\delta_1^\S, d_1} \cdots$
    of the acceptance game $\Game\S$.
    Consider the infinite play $\pi = \tuple{a_0, m_0, f_0, d_0} \tuple{a_1, m_1,f_1, d_1} \cdots$ of the separability game $\Game{}$.
    Since the run $\rho^\A$ is accepting,
    the infinite sequence of $\A$-transitions $\delta_0^\A \delta_1^\A \cdots$ is accepting.
    Thus, \eqref{eq:lor-to-direction} implies that $\pi\in\Win\A$.
    Therefore, the infinite path $b := \tuple{a_0,d_0}\tuple{a_1,d_1}\cdots$ belongs to $L_\A\subseteq\lang\D$
    and thus the corresponding infinite play $\tuple{\delta^\S_0, d_0}\tuple{\delta^\S_1, d_1} \cdots$ of the acceptance game $\Game\S$
    is winning for \Automaton, as required.
    This concludes the argument establishing $\lang \A \subseteq \lang \S$.

    It remains to show that $\lang \S \disjoint \lang \B$,
    which is the same as $\lang \B \subseteq \lang {\S^\mathrm c}$ for the complement game automaton.
    This follows directly from the construction above via the duality of the game $\Game {}$.
    Indeed, consider the generalised game automaton $\S$ as defined in \cref{sec:game_gen}. Let $\S^\mathrm c$ be the complementary game automaton $\tuple{\Sigma, M \cup \set \top, \ell_0, \Delta^{\S^\mathrm c}, \D^\mathrm c}$, which recognises the complement language of $\S$.
    (Here, $\Delta^{\S^\mathrm c}$ is the dualisation of $\Delta^\S$ as in the proof of \cref{rem:game-automata-dual},
    and $\D^\mathrm c$ is the complementary automaton to $\D$---its priorities are increased by $1$.)
    We first observe that if $\M=\tuple{ M, \ell_0, (\overline{m},\overline{f}), \tau}$ is a~strategy of \Separator in $\GameSeparabilityGame \A \B$,
    then $\M^{\mathrm c}:= \tuple{ M, \ell_0, (\overline{m}^{\mathrm c},\overline{f}), \tau}$ with $\overline{m}^{\mathrm c}$ returning the opposite mode than $\overline{m}$
    is a~strategy of \Separator in $\GameSeparabilityGame \B \A$. By the symmetry of the winning condition, $\M$ is winning if and only if $\M^{\mathrm c}$ is winning.
    The following claim follows directly from the definition of $\alpha$.
    \begin{claim}
        If $\M$ is a~winning strategy of \Separator in $\GameSeparabilityGame \A \B$ then
        \[\S^\mathrm c = \alpha(\M,\D)^{\mathrm c}= \alpha(\M^{\mathrm c}, \D^{\mathrm c}).\]
    \end{claim}
    Therefore, by applying the argument that $\lang\A\subseteq\lang \S$ to $\M^{\mathrm c}$, $\D^{\mathrm c}$ for the game $\GameSeparabilityGame \B \A$,
    we obtain $\lang \B \subseteq \lang {\S^\mathrm c}$, as required.
\end{proof}


\begin{proof}[Completeness]
    Assume that there exists a~game automaton $\S$ that separates $\lang \A$ from $\lang \B$.
    We need to show that \Separator wins the separability game $\Game{} := \GameSeparabilityGame \A \B$.
    Let $\T:=\S^{\mathrm c}$ be the syntactic dual of the game automaton $\S$ as in \cref{rem:game-automata-dual}. Thus, the automata $\S$ and $\T$ share the same set of states. Also, their transitions are related: the conjunctive transitions of $\S$ correspond to disjunctive transitions of $\T$ and vice versa. By slightly rephrasing the separation condition,
    we have $\lang \A \disjoint \lang \T$ and $\lang \B \disjoint \lang \S$. This means that \Pathfinder wins both disjointness games $\DisjointnessGame \T \A$ and $\DisjointnessGame \S \B$.
    Thus, we can apply \cref{cor:disjoint-pathfinder} to obtain pathfinders
    $\PP_\A\colon \left(\bigcup_{a\in\Sigma} \Delta^\T(a)\times \Delta^\A(a) \right)\to \set{\L,\R}$ and
    $\PP_\B\colon \left(\bigcup_{a\in\Sigma} \Delta^\S(a)\times \Delta^\B(a) \right)\to \set{\L,\R}$.

    We will now provide a~strategy of \Separator in $\Game{}$. The constructed strategy uses as its memory states the set of states of $\S$ that are distinct than $\top$. Let the initial memory state be $q_0$. Assume that the current memory state is $q_i$ and consider the $i$-th round of the game.

    \begin{simuize}
    \item[\Choice{I}{a}] \Input plays an~arbitrary letter $a_i \in\Sigma$.
    \item[\Choice{S}{m}] \Separator plays the mode $m_i\in\set{\lor,\land}$ defined as follows.
    We consider the following two cases for the mode of the transitions $\Delta^\S(q_i,a_i)$.

    \begin{enumerate}
    \item If $\Delta^\S(q_i,a_i)=\set{\delta^\S_i}$ is a~single conjunctive transition $\delta^\S_i=\tuple{q_i,a_i,q_{\L,i},\allowbreak q_{\R,i}}$ then we put $m_i:=\land$ and $f_i:= \PP_\B(\delta^\S_i,\_)$ is a~selector for $\B$.
    \item Otherwise, $\Delta^\S(q_i,a_i)$ is a~pair of disjunctive transitions which means that $\Delta^\T(q_i,a_i)$ is a~single conjunctive transition $\delta^\T_i=\tuple{q_i,a_i,q_{\L,i},q_{\R,i}}$. In this case we put $m_i:=\lor$ and $f_i:= \PP_\A(\delta^\T_i,\_)$ is a~selector for $\A$.
    \end{enumerate}

    \item[\Choice{S}{f}] \Separator plays the selector $f_i$ defined above (notice that $f_i$ is either a~selector for $\A$ or for $\B$, according to $m_i$).
    \item[\Choice{I}{d}] \Input plays an~arbitrary direction $d_i\in\set{\L,\R}$.
    \end{simuize}

    The next memory state of our strategy is the state $q_{d_i,i}$ taken from one of the transitions $\delta^\S_i$ or $\delta^\T_i$, see above. 
    We now argue that \Separator wins the corresponding infinite play
    $\pi = \tuple{a_0,m_0,f_0,d_0}\allowbreak \tuple{a_1,m_1,f_1,d_1} \cdots$.
    Let $b = \tuple{a_0,d_0}\tuple{a_1,d_1}\cdots$ be the corresponding path.
    Consider a~number $i\in\omega$. By the construction of the strategy above, we have two cases:
    \begin{enumerate}
    \item If $m_i=\land$, then a~conjunctive transition $\delta^\S_i=\tuple{q_i,a_i,q_{\L,i},q_{\R,i}}$ of $\S$ was used to determine $f_i$. In this case, define $\delta^\T_i$ as the following disjunctive transition of $\T$: if $d_i=\L$ then $\delta^\T_i:= \tuple{q_i,a_i,q_{\L,i},\top}$, otherwise $d_i=\R$ and $\delta^\T_i:= \tuple{q_i,a_i,\top,q_{\R,i}}$.
    \item If $m_i=\lor$, then a~conjunctive transition $\delta^\T_i=\tuple{q_i,a_i,q_{\L,i},q_{\R,i}}$ of $\T$ was used to determine $f_i$. In this case, define $\delta^\S_i$ as the following disjunctive transition of $\S$: if $d_i=\L$ then $\delta^\S_i:= \tuple{q_i,a_i,q_{\L,i},\top}$, otherwise $d_i=\R$ and $\delta^\S_i:= \tuple{q_i,a_i,\top,q_{\R,i}}$.
    \end{enumerate}
    The definitions above provide two sequences of transitions 
    $\vec\delta^\S:=\delta^\S_0\delta^\S_1\cdots\in\Delta^\S(b)$,
    $\vec\delta^\T:=\delta^\T_0\delta^\T_1\cdots \in\Delta^\T(b)$.
    Since for every $i\in\omega$ the transitions $\delta^\S_i$ and $\delta^\T_i$ are from the same state $q_i\neq\top$,
    $\vec\delta^\S$ is accepting in $\S$ if, and only if,
    $\vec\delta^\T$ is rejecting in $\T$.
    Assume that $\vec\delta^\S$ is accepting (the other case is analogous).
    We will show that $\Win\B$ is violated (if $\vec\delta^\T$ is accepting then $\Win\A$ is violated). Assume for the sake of contradiction that $\Win\B$ holds, as witnessed by a~sequence of $\B$-transitions $\vec\delta^\B = \delta_0^\B \delta_1^\B \cdots \in \Delta^\B(b)$.
    By \cref{rem:pathfinder-for-game} we obtain that whenever $m_i=\lor$ and $\delta^\S_i$ is a~disjunctive transition of $\S$ then $\PP_\B(\delta^\S_i,\_)$ is constantly equal to $d_i$.
    By the assumption on $\vec\delta^\B$ from $\Win\B$ we know that whenever $m_i=\land$ then $f_i(\delta^\B_i)=d_i$. However, if $m_i=\land$ then $f_i(\delta^\B_i)=\PP_\B(\delta^\S_i,\delta^\B_i)$. Therefore, in both cases we know that $\PP_\B(\delta^\S_i,\delta^\B_i)=d_i$.
    This means that the assumptions of \cref{rem:pathfinder-property} are met and at least one of the sequences $\vec\delta^\S$, $\vec\delta^\B$ is rejecting---a~contradiction, since we assumed both these sequences to be accepting.
\end{proof}

\section{Separability by game automata with priorities in \texorpdfstring{$C$}{C}}
\label{sec:game_ind}

In this section we present our last game-theoretic characterisation,
namely game automata separability for a fixed finite set $C\subseteq\N$ of priorities.
Fix two automata $\A = \tuple{\Sigma, Q^\A, q_0^\A, \Omega^\A, \Delta^\A}$ and
$\B = \tuple{\Sigma, Q^\B, q_0^\B, \Omega^\B, \Delta^\B}$ over the same alphabet $\Sigma$.
The game is a variation of $\GameSeparabilityGame \A \B$ from \cref{sec:game_gen}
where \Separator additionally plays priorities from $C$.
\begin{gamebox}{$C$-game-automata separability game $\CGameSeparabilityGame \A \B C$}
    At the $i$-th round:
    \begin{gameize}
        \item[\Choice{S}{c}] \Separator plays a priority $c_i \in C$.
        \item[\Choice{I}{a}] \Input plays a letter $a_i \in \Sigma$.
        \item[\Choice{S}{m}] \Separator plays a mode $m_i\in\set{\lor,\land}$.
        \item[\Choice{S}{f}] \Separator plays either
        \begin{enumerate}
            \item a selector $f_i \in \set{\L, \R}^{\Delta^\A(a_i)}$ for $\A$ if $m_i=\lor$, or
            \item a selector $f_i \in \set{\L, \R}^{\Delta^\B(a_i)}$ for $\B$ if $m_i=\land$.
        \end{enumerate}
        \item[\Choice{I}{d}] \Input plays a direction $d_i \in \set{\L, \R}$.
    \end{gameize}
\end{gamebox}
\noindent
\Separator wins an~infinite play $\pi = \tuple{c_0,a_0,m_0,f_0,d_0}\tuple{c_1,a_1,m_1,f_1,d_1}\cdots$
inducing a~path $b = \tuple{a_0,d_0}\tuple{a_1,d_1} \cdots$
whenever both conditions below hold:
\begin{winnize}
    \item  $\pi\in \Win\A$: If there exists an accepting sequence of transitions
    $\vec\delta^\A = \delta_0^\A \delta_1^\A \cdots \in \Delta^\A(b)$
    s.t.~for all $i\in\omega$ we have $(m_i = \lor) \Rightarrow f_i(\delta^\A_i) = d_i$,
    then $c_0c_1 \cdots$ is accepting.

    \item  $\pi\in \Win\B$: If there exists an accepting sequence of transitions
    $\vec\delta^\B = \delta_0^\B \delta_1^\B \cdots \in \Delta^\B(b)$
    s.t.~for all $i\in\omega$ we have $(m_i = \land)\Rightarrow f_i(\delta^\B_i) = d_i$,
    then $c_0c_1 \cdots$ is rejecting.
\end{winnize}

\begin{restatable}{lemma}{lemCGameSeparability}
    \label{lem:CGameSeparability:correctness}
    \Separator wins $\CGameSeparabilityGame \A \B C$ if, and only if,
    there exists a~game automaton $\S$ with priorities in $C$ separating $\lang \A$, $\lang \B$.
\end{restatable}

\noindent
The proof of this lemma can be seen as a~simplified variant of the proof of \cref{lem:game:sep},
except for the acceptance condition of the separator which is given by the priorities $c_i$'s
as in the proof of \cref{lem:CDeterministicSeparability:correctness}.
A strategy for \Separator in $\CGameSeparabilityGame \A \B C$ is a tuple
\begin{align}
    \label{eq:strategy:C-game-sep}
    \M = \tuple{M, \ell_0, \tuple{\overline{c},\overline{m},\overline{f}}, \tau}
\end{align}
where $M$ is a set of memory states,
$\ell_0 \in M$ is the initial memory state,
$\overline{c},\overline{m},\overline{f}$ are decision functions,
and $\tau : M \times \Sigma \times \set{\L, \R} \to M$ is the memory update function.
More precisely, $\overline c \colon M \to C$ outputs a priority $\overline c(\ell)$ in position $\ell$,
$\overline m \colon M \times \Sigma \to \set{\land, \lor}$ outputs a~mode $\overline m(\ell, a)$ in position $\ell$ when \Input plays $a \in \Sigma$,
and $\overline f \colon M \times \Sigma \to \bigcup_{a \in \Sigma} (\set{\L, \R}^{\Delta^\A(a)} \cup \set{\L, \R}^{\Delta^\B(a)})$ outputs a selector
$\overline f(\ell, a)$ in similar circumstances.
With this notation we can define a~correspondence $\alpha$ from finite-memory winning strategies for \Separator in $\CGameSeparabilityGame \A \B C$ to game automata separating $\A$, $\B$ with priorities in $C$.
More precisely, we map an arbitrary finite-memory strategy $\M$ to a game automaton
\begin{align}
    \S := \alpha(\M) := \tuple{\Sigma, M, \ell_0, \Delta^\S, \Omega^\S}
\end{align}
which has the same set of states $M$ and initial state $\ell_0$ as $\M$,
priorities are induced by the decision function $\overline{c}$ of $\M$ as
\begin{align*}
    \Omega^\S(\ell) := \overline{c}(\ell),
\end{align*}
and transitions are of the form
	\[
        \Delta^\S(\ell, a) =
			\begin{cases}
			    \set{\tuple{\ell,a,\ell_\L, \top}, \tuple{\ell,a,\top,\ell_\R}}
                & \text{if $\overline{m}(\ell, a) = \lor $,}\\
            	\set{\tuple{\ell,a,\ell_\L, \ell_\R}}
                & \text{if $\overline{m}(\ell, a) = \land$,}
			\end{cases}
    \]
where $\ell_\L := \tau(\ell, a, \L)$ and $\ell_\R := \tau(\ell, a, \R)$.
Notice how the acceptance condition of $\S$ is simply inherited from the winning strategy $\M$.
This should be contrasted with \cref{sec:game_gen} where the set of priorities $C$ is not fixed beforehand,
and thus the acceptance condition of $\S$ is defined with the help of the winning condition for \Separator in the corresponding separability game.
The decision function $\overline f$ is not involved in the definition of $\alpha(\M)$,
however it is used to show that if $\M$ is winning, then $\alpha(\M)$ is in fact a~separator.

\renewcommand{\SepGame}{\Game{}}
\renewcommand{\AccGameA}{{\color{darkgreen}G_\A}}
\renewcommand{\AccGameS}{{\color{darkblue}G_\B}}
In the following, let $\SepGame := \CGameSeparabilityGame \A \B C$ be the $C$-game-separability game.
The proof below is very similar to the one of \cref{lem:game:sep},
with some adaptations to take care of the additional priorities $c_i$'s selected by \Separator.

\begin{proof}[Soundness]
    Assume that \Separator wins the separability game $\SepGame$.
    By \cref{thm:BuchiLandweber},
    there exists a~finite-memory winning strategy $\M$ as in \eqref{eq:strategy:C-game-sep}.
    Let $\S = \alpha(\M)$ be the game automaton corresponding to $\M$.
    %
    %
    We show that $\S$ separates $\lang \A$ from $\lang \B$,
    i.e.,~$\lang \A \subseteq \lang \S$ and $\lang \S \disjoint \lang \B$.
    First, we show that $\lang \A \subseteq \lang \S$.
    To this end, assume that $t \in \lang \A$ and let $\rho^\A$ be an~accepting run of $\A$ over $t$ witnessing this.
%
%
	We show that \Automaton wins the acceptance game $\Game\S := \AcceptanceGame \S t$.
	To show this we play in parallel the separability game $\Game{}$ and the acceptance game $\Game\S$.
	As we play both games in lock-steps, we maintain the following invariant:
	At every round $i$, the current finite path of the input tree $t$ is $\tuple{a_0,d_0} \cdots \tuple{a_{i-1}, d_{i-1}}$,
	\Separator's winning strategy $\M$ in the separability game $\Game{}$ is in memory state $\ell_i$,
	the current state of the separating automaton $\S$ in the acceptance game $\Game\S$ is $\ell_i$ as well,
	and $\rho^\A(d_0\cdots d_{i-1})=q_i^\A$.
    Let now be at round $i$ and assume that the invariant holds.
    We play the separability and acceptance games as follows.
\begin{simuize}
\item[\GChoice{\Game{}}{S}{c}] \Separator plays the~priority $c_i:= \overline{c}(\ell_i) \in C$.
\item[\GChoice{\Game{}}{I}{a}] \Input plays the letter $a_i := t(u_i)$ for $u_i:=d_0 \cdots d_{i-1}$.
\item[\GChoice{\Game{}}{S}{m}] \Separator plays the~mode $m_i:= \overline{m}(\ell_i,a_i) \in\set{\lor,\land}$.
\item[\GChoice{\Game{}}{S}{f}] \Separator plays either
        \begin{enumerate}
            \item the selector $f_i :=\overline{f}(\ell_i,a_i) \in \set{\L, \R}^{\Delta^\A(a_i)}$ for $\A$ if $m_i=\lor$ or
            \item the selector $f_i :=\overline{f}(\ell_i,a_i) \in \set{\L, \R}^{\Delta^\B(a_i)}$ for $\B$ if $m_i=\land$.
        \end{enumerate}
    \item[\GChoice{\Game\S}{A}{\delta}] \Automaton plays the transition $\delta^\S_i\in\Delta^\S(\ell_i, a_i)$, defined as follows.
    Let $\delta^\A_i:= \big(\rho^\A(u_i), t(u_i),\rho^\A(u_i\L),\rho^\A(u_i\R)\big)$ be the $\A$-transition used in $u_i$ by the run $\rho^\A$. We distinguish two cases.

    \begin{enumerate}
    \item In the first case, assume that \Separator played $m_i=\lor$ and $f_i \in \set{\L, \R}^{\Delta^\A(a_i)}$.
    It means that $\Delta^\S\big((\ell_i,q_i), a_i\big)$ contains two disjunctive transitions,
    $\delta^\S_{\L,i}:=\tuple{\ell_i,a_i,\ell_{\L,i},\top}$ and
    $\delta^\S_{\R,i}:=\tuple{\ell_i,a_i,\top,\ell_{\R,i}}$. Let us put $\delta^\S_i:=\delta^\S_{f_i(\delta^\A_i), i}$, i.e.,~the transition that sends a~non-$\top$ state in the direction given by $f_i(\delta^\A_i)$.
    
    \item In the second case, \Separator played $m_i=\land$ and $f_i \in \set{\L, \R}^{\Delta^\B(a_i)}$.
    It means that $\Delta^\S\big(\ell_i, a_i\big)$ contains one conjunctive transition $\delta^\S_i:=\tuple{\ell_i,a_i,\ell_{\L,i},\ell_{\R,i}}$.
    \end{enumerate}

    \item[\GChoice{\Game\S}{I}{d}] \Input plays an~arbitrary direction $d_i\in\set{\L,\R}$.
    \item[\GChoice{\Game{}}{I}{d}] \Input plays the direction $d_i\in\set{\L,\R}$.
\end{simuize}

Notice that if $m_i = \lor$ and $d_i\neq f^\lor(\delta_i^\A)$
then the next position of the acceptance game $\Game\S$ is $(u_id_i, \top)$, which is a~winning position for \Automaton.
Therefore, without loss of generality we can assume that
\begin{equation}
\label{eq:direction-choice-gam-ind}
    \forall i\in\omega.\ (m_i = \lor) \Rightarrow f^\lor(\delta_i^\A)=d_i.
\end{equation}
Moreover, the new state of $\S$ in $\Game\S$ is $\ell_{i+1}$ for $\ell_{i+1} := \tau(\ell_i, a_i, d_i)$.
Similarly, the new memory state of $\M$ in $\Game{}$ is $\ell_{i+1}$.
This concludes the description of round $i$ of both games.
We argue that \Automaton wins the resulting infinite play
$\tuple{\delta_0^\S, d_0}\tuple{\delta_1^\S, d_1} \cdots$
of the acceptance game $\Game\S$.
Consider the infinite play $\pi = \tuple{c_0, a_0, m_0, f_0, d_0} \tuple{c_1, a_1, m_1,f_1, d_1} \cdots$ of the separability game $\Game{}$.
Let $b:= \tuple{a_0, d_0} \tuple{a_1, d_1} \cdots$ be the induced path.
Since we used the winning strategy of \Separator, this play satisfies $\Win\A$.
Since the run $\rho^\A$ is accepting,
the infinite sequence of $\A$'s transitions $\delta_0^\A \delta_1^\A \cdots$ is accepting.
Additionally, \eqref{eq:direction-choice-gam-ind} holds. Therefore, the sequence of priorities $c_0c_1\cdots$ must be accepting by $\Win\A$. However, by the definition of the automaton $\S$, we know that $\Omega^\S(\ell_i)=c_i$, which means that \Automaton wins the considered play of the acceptance game $\Game\S$.

    It remains to prove that $\lang S \disjoint \lang \B$.
    The latter is equivalent to $\lang \B \subseteq \lang {\S^\mathrm c}$,
    where the dual game automaton $\S^\mathrm c$ recognises the complement language of $\S$.
    By the construction, $\S^\mathrm c$ has the same states as $\S$,
    and transitions are defined by exchanging the conjunctive and disjunctive ones.
    Moreover, the priorities in $\S^\mathrm c$ can be chosen as $\Omega^{\S^\mathrm c}(\ell) = \Omega^\S(\ell) + 1$,
    and thus a sequence of priorities $\Omega^\S(\ell_0)\Omega^\S(\ell_1)\cdots$ in $\S$ is rejecting if, and only if,
    the corresponding sequence $\Omega^{\S^\mathrm c}(\ell_0)\Omega^{\S^\mathrm c}(\ell_1)\cdots$ in $\S^\mathrm c$ is accepting.
    With these observations in hand, we can conclude by repeating the argument in the first part of the proof above
    with $\A$ replaced by $\B$,
    $\S$ replaced by $\S^\mathrm c$,
    and condition $\Win\A$ replaced by $\Win\B$.
\end{proof}

\newcommand{\DisGameSA}{{\color{darkgreen}G_{\S^\mathrm c, \A}}}
\newcommand{\DisGameSB}{{\color{darkblue}G_{\S, \B}}}

\begin{proof}[Completeness]
    Assume that $\S$ is a~game automaton with priorities in $C$ separating $\lang \A$ from $\lang \B$,
    and we show that \Separator wins the separability game $\SepGame:=\CGameSeparabilityGame \A \B C$.
    Let $\T:=\S^{\mathrm c}$ be the dual game automaton
    recognising the complement language $\lang\T = \trees_\Sigma \setminus \lang \S$.

    Since $\S$ is a separator,
    we have that $\lang \S \disjoint \lang \B$ and $\lang \T \disjoint \lang \A$, which means that \Pathfinder wins both
	disjointness games $\DisjointnessGame \S \B$ and $\DisjointnessGame \T \A$. Let
\begin{align*}
\PP_\B\colon &\left(\bigcup_{a\in\Sigma} \Delta^\S(a)\times \Delta^\B(a)\right)\to \set{\L,\R},\\
\PP_\A\colon &\left(\bigcup_{a\in\Sigma} \Delta^\T(a)\times \Delta^\A(a)\right)\to \set{\L,\R},
\end{align*}
be two pathfinders witnessing this.

We will now provide a~strategy of \Separator in $\Game{}$.
The constructed strategy uses as its memory states the set of states of $\S$ that are distinct than $\top$. Let the initial memory state be $q_0$. Assume that the current memory state is $q_i$ and consider the $i$-th round of the game.

\begin{simuize}
\item[\Choice{S}{m}] \Separator plays the priority $c_i:=\Omega^\S(\ell_i)\in\set{\lor,\land}$.
\item[\Choice{I}{a}] \Input plays an arbitrary letter $a_i \in\Sigma$.
\item[\Choice{S}{m}] \Separator plays the mode $m_i\in\set{\lor,\land}$ defined as follows.
We consider the following two cases for the mode of the transitions $\Delta^\S(q_i,a_i)$.

\begin{enumerate}
\item If $\Delta^\S(q_i,a_i)=\set{\delta^\S_i}$ is a~single conjunctive transition $\delta^\S_i=\tuple{q_i,a_i,q_{\L,i},\allowbreak q_{\R,i}}$ then we put $m_i:=\land$ and $f_i:= \PP_\B(\delta^\S_i,\_)$ is a~selector for $\B$.
\item Otherwise, $\Delta^\S(q_i,a_i)$ is a~pair of disjunctive transitions which means that $\Delta^\T(q_i,a_i)$ is a~single conjunctive transition $\delta^\T_i=\tuple{q_i,a_i,q_{\L,i},q_{\R,i}}$. In this case we put $m_i:=\lor$ and $f_i:= \PP_\A(\delta^\T_i,\_)$ is a~selector for $\A$.
\end{enumerate}

\item[\Choice{S}{f}] \Separator plays the selector $f_i$ defined above (notice that $f_i$ is either a~selector for $\A$ or for $\B$, according to $m_i$).
\item[\Choice{I}{d}] \Input plays an arbitrary direction $d_i\in\set{\L,\R}$.
\end{simuize}

The next memory state of our strategy is the state $q_{d_i,i}$ taken from one of the transitions $\delta^\S_i$ or $\delta^\T_i$, see above.

We now argue that \Separator wins the corresponding infinite play
$\pi = \tuple{a_0,m_0,f_0,d_0}\allowbreak \tuple{a_1,m_1,f_1,d_1} \cdots$.
Let $b = \tuple{a_0,d_0}\tuple{a_1,d_1}\cdots$ be the corresponding path.

We begin by showing that $\pi \in \Win\B$.
    Let
    \begin{align}
        \label{eq:delta:B}
        \vec\delta^\B = \delta^\B_0 \delta^\B_1 \cdots \in \Delta^\B(b)
    \end{align}
    be an~infinite accepting sequence of transitions over the branch $b$ conform to $\pi_0$,
    where $\delta^\B_i$ has the form $\delta^\B_i = \tuple{q^\B_i, a_i, q^\B_{\L, i}, q^\B_{\R, i}}$.
    We need to show that $c_0c_1\cdots$ is rejecting.

Consider a~number $i\in\omega$. By the construction of the strategy of \Separator above, we know that there are two cases:
\begin{enumerate}
\item If $m_i=\land$, then a~conjunctive transition $\delta^\S_i=\tuple{q_i,a_i,q_{\L,i},q_{\R,i}}$ of $\S$ was used to determine $f_i$. In this case, define $\delta^\T_i$ as the following disjunctive transition of $\T$:
If $d_i=\L$, then $\delta^\T_i:= \tuple{q_i,a_i,q_{\L,i},\top}$,
otherwise $d_i=\R$ and $\delta^\T_i:= \tuple{q_i,a_i,\top,q_{\R,i}}$.
\item If $m_i=\lor$, then a~conjunctive transition $\delta^\T_i=\tuple{q_i,a_i,q_{\L,i},q_{\R,i}}$ of $\T$ was used to determine $f_i$. In this case, define $\delta^\S_i$ as the following disjunctive transition of $\S$:
If $d_i=\L$, then $\delta^\S_i:= \tuple{q_i,a_i,q_{\L,i},\top}$, otherwise $d_i=\R$ and $\delta^\S_i:= \tuple{q_i,a_i,\top,q_{\R,i}}$.
\end{enumerate}
The definitions above provide two sequences of transitions: $\vec\delta^\S:=\delta^\S_0\delta^\S_1\cdots\in(\Delta^\S)^\omega$ and $\vec\delta^\T:=\delta^\T_0\delta^\T_1\cdots \in (\Delta^\T)^\omega$.
Notice that the construction guarantees that $\vec\delta^\S\in\Delta^\S(b)$ and $\vec\delta^\T\in\Delta^\T(b)$.

By \cref{rem:pathfinder-for-game} we obtain that if $m_i=\lor$ and $\delta^\S_i$ is a~disjunctive transition of $\S$,
then $\PP_\B(\delta^\S_i,\_)$ is constantly equal $d_i$ (the direction in which $\delta^\S_i$ sends the state different than $\top$).
By the assumption on $\vec\delta^\B$ from $\Win\B$ we know that if $m_i=\land$,
then $f_i(\delta^\B_i)=d_i$.
However, if $m_i=\land$, then $f_i(\delta^\B_i)=\PP_\B(\delta^\S_i,\delta^\B_i)$. Therefore, in both cases we know that $\PP_\B(\delta^\S_i,\delta^\B_i)=d_i$.

This means that the assumptions of \cref{rem:pathfinder-property} are met and at least one of the sequences $\vec\delta^\S$, $\vec\delta^\B$ is rejecting. Since we assumed that $\vec\delta^\B$ is accepting, $\vec\delta^\S$ must be rejecting. But the priorities $c_0c_1\cdots$ are just the priorities of the transitions $\delta^\S_i$, so $c_0c_1\cdots$ is rejecting.

The case of $\Win\A$ is entirely dual: we consider a~sequence of transitions $\vec\delta^\A = \delta^\A_0 \delta^\A_1 \cdots \in \Delta^\A(b)$ that is accepting and use \cref{rem:pathfinder-property} for $\PP_\A$ to show that $\vec\delta^\T$ must be rejecting, which implies that $c_0c_1\cdots$ is accepting.
\end{proof}

\section{Complexity}
\label{ap:complexity}

In this section we perform a detailed analysis of the complexity of solving the separability problems from \cref{sec:det_ind,sec:det_gen,sec:game_gen,sec:game_ind} and the complexity of separators,
thus proving \cref{thm:complexity} announced in the introduction:
\thmComplexity*
In each case it will be a matter of constructing a deterministic parity automaton $\WW$ over $\omega$-words recognising the set of winning plays and then solving a suitable parity game.
In the following, let $M = \set{\lor, \land}$ be the set of alternation modes,
and let $D = \set{\L, \R}$ be the set of directions.

\subsection{Separability by deterministic automata}
\label{ap:complexity:det_gen}

In this section we perform a complexity analysis for \cref{sec:det_gen}.
Let $\A = \tuple{\Sigma, Q^\A, q_0^\A, \Omega^\A, \Delta^\A}$
and recall that $\Win\A$ is the set of plays
$\pi = \tuple{a_0,f_0,d_0}\allowbreak \tuple{a_1,f_1,d_1}\cdots$
with branch $b = \tuple{a_0, d_0} \tuple{a_1, d_1} \cdots$
s.t.~there is an accepting sequence of transitions
$\vec\delta^\A \in \Delta^\A(b)$.
The language $\Win\A$ can be recognised by a nondeterministic $\omega$-word parity automaton $\WW_\A$
over the alphabet 
\begin{align}
    \label{eq:Sigma':det}
    \Sigma' = \Sigma \times (\bigcup_{a \in \Sigma} D^{\Delta^\B(a)}) \times D.
\end{align}
which reads $\pi$, nondeterministically guesses the sequence of transitions $\vec\delta^\A$,
and verifies that it is accepting.
(Notice that $\Sigma'$ has size exponential in the size of $\A$.)
More precisely, we can take $\WW_\A = \tuple{\Sigma', Q^\A, q_0^\A, \Omega^\A, \Delta^{\WW_\A}}$
to have the same states $Q^\A$, initial state $q_0^\A$, and priority function $\Omega^\A$ as $\A$,
and set of transitions
\begin{align*}
    \Delta^{\WW_\A} = \setof{\tuple{q, a', q_d}}{\text{for some } a' = \tuple{a, f, d} \in \Sigma' \text{ and } \delta^\A = \tuple{q, a, q_\L, q_\R} \in \Delta^\A}.
\end{align*}
It is immediate to verify that $\Win\A = \lang {\WW_\A}$.
Let $\B = \tuple{\Sigma, Q^\B, q_0^\B, \Omega^\B, \Delta^\B}$
and recall that $\Win\B$ is the set of plays $\pi$ with branch $b$ as above
s.t.~there is an accepting sequence of transitions
$\vec\delta^\B = \delta_0^\B \delta_1^\B \cdots \in \Delta^\A(b)$
s.t., for all $i \in \N$, $f_i(\delta_i^\A) = d_i$.
As above, the language $\Win\B = \lang {\WW_\B}$ can be recognised by a nondeterministic $\omega$-word parity automaton
$\WW_\B = \tuple{\Sigma', Q^\B, q_0^\B, \Omega^\B, \Delta^{\WW_\B}}$
where
\begin{align*}
    \Delta^{\WW_\B} = \Setof{\tuple{q, a', q_d}}{
        \begin{array}{c}
            \text{for some } a' = \tuple{a, f, d} \in \Sigma'
                \text{ and } \delta^\B = \tuple{q, a, q_\L, q_\R} \in \Delta^\B\\
            \text{ s.t.~} f(\delta^\B) = d.
        \end{array}
        }
\end{align*}
Putting the two constructions above together,
\Input's winning condition $\Win\Input = \Win\A \cap \Win\B$
can be recognised by a nondeterministic $\omega$-word parity automaton of size polynomial in $\A, \B$,
and thus by \cref{lem:NPA2DPA} by a deterministic $\omega$-word parity automaton $\WW_\Input$
of exponential size and polynomially many priorities.
By applying \cref{lem:parity:games} and the characterisation of \cref{lem:det:sep}
we can thus solve the deterministic separability problem in \EXPTIME.
Thanks to the implication ``$3 \Rightarrow 2$'' of \cref{lem:det:sep},
if a deterministic separator exists,
then the path closure automaton $\pathautomaton \A$ is a deterministic separator.
By inspecting the construction of $\pathautomaton \A$,
one can see that it has number of states exponential in that of $\A$,
and the same set of priorities as $\A$.
This discussion is summarised in the following result.

\begin{theorem}
    \label{thm1}
    The deterministic separability problem can be solved in \EXPTIME.
    Moreover, when a deterministic separator exists,
    there is one with exponentially many states and polynomially many priorities.
\end{theorem}

\subsection{Separability by deterministic automata with priorities in \texorpdfstring{$C$}{C}}
\label{sec:complexity:C-det-sep}

In this section we perform a complexity analysis for \cref{sec:det_ind}.
We build a nondeterministic automaton $\WW_\A = \tuple{\Sigma'', Q, q_0, \Omega, \Delta}$
recognising the set of plays $\lang {\WW_\A}$ \emph{not} satisfying $\Win\A$.
Automaton $\WW_\A$ is over the alphabet
\begin{align}
    \label{eq:Sigma'':det}
    \Sigma'' = C \times \Sigma',
\end{align}
with $\Sigma'$ from \eqref{eq:Sigma':det}.
Intuitively, $\WW_\A$ accepts an infinite play
$\pi = \tuple{c_0, a_0, f_0, d_0} \allowbreak\tuple{c_1, a_1, f_1, d_1} \cdots$
with path $b = \tuple{a_0, d_0} \allowbreak \tuple{a_1, d_1} \cdots$
whenever there exists an accepting sequence of transitions
$\vec \delta^\A = \delta_0^\A \delta_1^\A \cdots \in \Delta^\A(b)$
and $c_0c_1 \cdots$ is rejecting.
In order to achieve this, $\WW_\A$ guesses an accepting sequence of transitions from $\A$ (as in \cref{ap:complexity:det_gen})
and also guesses an odd priority $c \in C$ and verifies that it occurs infinitely often,
and that no larger priority occurs infinitely often.
This can be achieved by a set of states $Q$ of size polynomial in $\A$.
Note that the input alphabet $\Sigma''$ has exponential size in $\A, \B$ (due to the selectors $f_i$'s),
and thus $\WW_\A$ will have exponentially many transitions.
A very similar construction yields a nondeterministic parity $\omega$-word automaton $\WW_\B$
over the same action alphabet $\Sigma''$ from \eqref{eq:Sigma'':det}
recognising the set of plays $\lang {\WW_\B}$ \emph{not} in $\Win\B$
%
with polynomially many states and exponentially many transitions.
It follows that the complement of $\Win\A \cap \Win\B$ can be recognised
by a nondeterministic parity $\omega$-word automaton $\WW$
of the same complexity.
By \cref{lem:NPA2DPA} we can further convert $\WW$
to an equivalent deterministic parity automaton $\WW'$ with exponentially many states and polynomially many priorities
(w.r.t.~the number of states of $\A, \B$).
By \cref{lem:parity:games} we can thus solve $\CDeterministicSeparabilityGame \A \B C$ in \EXPTIME,
and by the characterisation in \cref{lem:CDeterministicSeparability:correctness}
we can solve the $C$-deterministic separability problem within the same complexity.

Based on the size of the winning condition $\WW'$
and the strong connection between winning strategies for \Separator and deterministic separators
in the ``soundness'' direction of the proof of \cref{lem:CDeterministicSeparability:correctness},
we can also provide an upper bound on the size of a separating deterministic automaton, when it exists.
More precisely, if \Separator wins the $C$-deterministic-separability game $\CDeterministicSeparabilityGame \A \B C$,
then she has a positional winning strategy in the corresponding graph game of exponential size from \cref{lem:parity:games}.
This means that \Separator has a winning strategy $\M$ of exponential memory in $\CDeterministicSeparabilityGame \A \B C$.
This strategy is then translated to a separating deterministic automaton $\S$ with exponentially many states and priorities in $C$.
Putting these considerations together gives the following complexity result.

\begin{theorem}
    \label{thm2}
    The $C$-deterministic separability problem is in \EXPTIME.
    Moreover, deterministic separators of exponential size suffice.
\end{theorem}

\subsection{Separability by game automata}

In this section we perform a complexity analysis for \cref{sec:game_gen}.
Let $\A = \tuple{\Sigma, P, p_0, \Omega, \Delta^\A}$
and recall that $\Win\A$ is the set of plays of the form
$\pi = \tuple{a_0,m_0,f_0,d_0}\allowbreak\tuple{a_1,m_0,f_1,d_1}\cdots$
with branch $b = \tuple{a_0, d_0}\allowbreak \tuple{a_1, d_1} \cdots$
s.t.~there is an accepting sequence of transitions
$\vec\delta^\A = \delta_0^\A \delta_1^\A \cdots \in \Delta^\A(b)$
s.t., for all $i \in \N$, (\dag) if $m_i = \lor$ then $f_i(\delta_i^\A) = d_i$.
The language $\Win\A$ can be recognised by a~nondeterministic $\omega$-word parity automaton $\WW_\A$
over the alphabet
\begin{align}
    \label{eq:Sigma'}
    \Sigma' = \Sigma \times M \times (\bigcup_{a \in \Sigma} D^{\Delta^\A(a)} \cup D^{\Delta^\B(a)}) \times D
\end{align}
which reads $\pi$, nondeterministically guesses the sequence of transitions $\vec\delta^\A$,
and verifies that it is accepting and that (\dag) defined above holds.
(Notice that $\Sigma'$ has size exponential in the size of $\A$.)
More precisely, we can take $\WW_\A = \tuple{\Sigma', P, p_0, \Omega, \Delta}$
to have the same states $P$, initial state $p_0$, and priority function $\Omega$ as $\A$,
and set of transitions
\begin{align*}
    \Delta = \setof{\tuple{q, a', p_d}}{&\text{for some } a' = \tuple{a, m, f, d} \in \Sigma' \text{ and }\delta^\A = \tuple{q, a, p_\L, p_\R} \in \Delta^\A\\
 &\text{ s.t.~} (m = \lor) \Rightarrow f(\delta^\A) = d}.
\end{align*}
It is immediate to verify that $\Win\A = \lang {\WW_\A}$.
With an analogous construction starting from $\B$
we can build a nondeterministic $\omega$-word parity automaton $\WW_\B$
recognising the set of plays in $\Win\B = \lang {\WW_\B}$.
Putting the two together, \Input's winning condition $\Win\Input = \Win\A \cap \Win\B$
can be recognised by a nondeterministic $\omega$-word parity automaton of polynomially many states and exponentially many transitions w.r.t.~$\A, \B$,
and thus by \cref{lem:NPA2DPA} by a deterministic $\omega$-word parity automaton $\WW_\Input$
of exponential size and polynomially many priorities.
By \cref{lem:parity:games} we can solve such a game in \EXPTIME,
and thanks to the characterisation from \cref{lem:game:sep},
we can solve the game separability problem in \EXPTIME.

In fact, we can also provide an upper bound on the number of states and priorities
of a separating game automaton (when it exists).
Since parity games are memoryless determined and the graph game has exponential size,
if \Separator wins $\GameSeparabilityGame \A \B$
then she has a winning strategy $\M$ of exponential memory.
This means that the separating automaton $\alpha(\M, \D)$ with generalised acceptance condition $\D$
has exponential size (ignoring the size of $\D$ for a moment).

We now argue about the size of a suitable deterministic automaton $\D$ for the generalised acceptance condition.
First of all, the winning condition $\Win\A \subseteq (\Sigma')^\omega$
is recognised by the nondeterministic parity automaton $\WW_\A$ above
with the same number of states as $\A$ and exponentially many transitions (since $\Sigma'$ has exponential size).
As suggested in the ``soundness'' direction of the proof of \cref{lem:game:sep},
we take $\D$ to be a deterministic automaton recognising the language $L_\A \subseteq (\Sigma \times D)^\omega$
containing all paths $b = \tuple{a_0, d_0}\tuple{a_1, d_1} \cdots$
s.t.~there exists a play $\pi \in \Win\A$ conform to $b$ and \Separator's strategy $\M$.
The automaton $\D$ can be obtained as a product construction
of $\WW_\A$ (polynomial) above and \Separator's strategy
$\M = \tuple{L, \ell_0, \tuple{\overline{c},\overline{m},\overline{f}}, \tau}$ (exponential),
a projection operation from alphabet $\Sigma'$ to alphabet $\Sigma \times D$,
and then a determinisation operation.
More precisely, let $$\D_0 = \tuple{\Sigma \times D, P \times L, \tuple{p_0, \ell_0}, \Omega_0, \Delta_0}$$
be a nondeterministic parity automaton over alphabet $\Sigma \times D$
where $\Delta_0$ and $\Omega_0$ are defined as follows:
$\tuple{\tuple{p, \ell}, \tuple{a, d}, \tuple{p', \ell'}} \in \Delta_0$
iff $\tuple{p, \tuple{a, \overline m(\ell, a), \overline f(\ell, a), d}, p'} \in \Delta$
and $\tau(\ell, a, d) = \ell'$;
$\Omega_0(p, \_) = \Omega(p)$.
Since $\M$ is winning and by the definition of $\WW_\A$ we have $\lang D_0 = L_\A$.
However $\D_0$ is nondeterministic and a direct determinisation seems to produce a doubly exponential blow-up
(since $L$ has exponential size).
However, the $L$-component of the state is in fact a deterministic finite automaton (with no acceptance condition),
and since the determinisation operation commutes with products with deterministic finite automata,
$\D_0$ can be determinised into an equivalent deterministic parity automaton $\D$
of exponential size and polynomially many priorities, as required.
%
%
By \cref{lem:generalised:tree:automata} applied to the generalised automaton $\alpha(\M, \D)$ 
we can build a game parity automaton $\S$ equivalent to $\alpha(\M, \D)$ (and thus separating $\lang \A, \lang \B$)
of exponential size and polynomially many priorities.
This discussion is summarised in the following result.

\begin{theorem}
    \label{thm3}
    The game separability problem for can be solved in \EXPTIME.
    Moreover, if a separating game automaton exists,
    then there is one with exponentially many states and polynomially many priorities.
\end{theorem}

\subsection{Separability by game automata with priorities in \texorpdfstring{$C$}{C}}

In this section we perform a complexity analysis for \cref{sec:game_ind}.
As in \cref{sec:complexity:C-det-sep}
one can build a nondeterministic parity automaton automaton $\WW_\A = \tuple{\Sigma'', Q, q_0, \Omega, \Delta}$
over alphabet $\Sigma'' = C \times \Sigma'$ (where $\Sigma'$ has been defined in \eqref{eq:Sigma'})
recognising the set of plays $\lang {\WW_\A}$ \emph{not} satisfying $\Win\A$
with polynomially many states $Q$ and priorities
and exponentially many transitions $\Delta$ (due to the exponential alphabet $\Sigma''$).
%
%
In the same way, we can build a nondeterministic parity $\omega$-word automaton $\WW_\B$
recognising the complement of the winning condition $\Win\B$,
and thus the complement of $\Win\A \cap \Win\B$ can be recognised
by a nondeterministic parity $\omega$-word automaton $\WW$
with polynomially many states and priorities and exponentially many transitions.
By \cref{lem:NPA2DPA} we can further convert $\WW$
to an equivalent deterministic parity automaton $\WW'$ with exponentially many states and polynomially many priorities
(w.r.t.~the number of states of $\A, \B$).
By \cref{lem:parity:games} we can thus solve $\CGameSeparabilityGame \A \B C$ in \EXPTIME,
and by the characterisation in \cref{lem:CGameSeparability:correctness}
we can solve the $C$-game separability problem within the same complexity.

Based on the size of the winning condition $\WW'$
and the strong connection between winning strategies for \Separator and separating automata
in the ``soundness'' direction of the proof of \cref{lem:CGameSeparability:correctness},
we can also provide an upper bound on the size of a separating game automaton, when it exists.
More precisely, if \Separator wins the $C$-game-separability game $\CGameSeparabilityGame \A \B C$,
then she has a positional winning strategy in the corresponding graph game of exponential size from \cref{lem:parity:games}.
This means that \Separator has a winning strategy $\M$ of exponential memory in $\CGameSeparabilityGame \A \B C$.
This strategy is then translated to a~separating game automaton $\alpha(\M)$ with exponentially many states and priorities in $C$.
Putting these considerations together gives the following complexity result.

\begin{theorem}
    \label{thm4}
    The $C$-game separability problem can be solved in \EXPTIME.
    Moreover, if a separating game automaton exists,
    then there exists one of exponential size.
\end{theorem}

Altogether, \cref{thm1,thm2,thm3,thm4} prove \cref{thm:complexity} announced in the introduction.

\bibliography{bibliography}

\end{document}